\documentclass[10pt,journal,compsoc]{IEEEtran}
\usepackage{graphicx}
\usepackage{subfigure}
\usepackage{amsmath}
\usepackage{amsthm}
\usepackage{amssymb}
\usepackage{mathrsfs}
\usepackage{slashbox}
\usepackage{float}
\usepackage[vlined,ruled,linesnumbered]{algorithm2e}
\usepackage{xcolor}
\usepackage{verbatim}

\newtheorem{thm}{Theorem}

\newtheorem{lem}{Lemma}

\newtheorem{definition}{Definition}

\newcount\DraftStatus  
\usepackage{color}
\definecolor{darkgreen}{rgb}{0,0.5,0}
\definecolor{purple}{rgb}{1,0,1}
\newcommand{\draftnote}[2]{\ifnum\DraftStatus=1
	\marginpar{
		\tiny\raggedright
		\hbadness=10000
        \def\baselinestretch{0.8}
        \textcolor{#1}{\textsf{\hspace{0pt}#2}}}
     \fi}

\begin{document}

\title{Differential Advising in Multi-Agent Reinforcement Learning}

\author{Dayong Ye, Tianqing Zhu*, Zishuo Cheng, Wanlei Zhou and Philip S. Yu
\thanks{*Tianqing Zhu is the corresponding author. T. Zhu is with the School of Computer Science, China University of Geosciences, Wuhan, P. R. China. D. Ye, Z. Cheng and W. Zhou are with the Centre for Cyber Security and Privacy and with the School of Computer Science, University of Technology, Sydney, Australia. Philip S. Yu is with Department of Computer Science, University of Illinois at Chicago, USA. Email: tianqing.e.zhu@gmail.com, \{Dayong.Ye, 13367166, Wanlei.Zhou\}@uts.edu.au, psyu@uic.edu. 
This work is supported by National Natural Science Foundation of China (No.61972366), ARC LP170100123 and DP190100981, Australia, 
and NSF under grants III-1763325, III-1909323, and SaTC-1930941, USA.}}

\maketitle
\pagestyle{empty}
\thispagestyle{empty}

\begin{abstract}
Agent advising is one of the main approaches to improve agent learning performance
by enabling agents to share advice.
Existing advising methods have a common limitation
that an adviser agent can offer advice to an advisee agent
only if the advice is created in the same state as the advisee's concerned state.
However, in complex environments, it is a very strong requirement that two states are the same,
because a state may consist of multiple dimensions and
two states being the same means that all these dimensions in the two states are correspondingly identical.
Therefore, this requirement may limit the applicability of existing advising methods to complex environments.
In this paper, inspired by the differential privacy scheme, we propose a differential advising method
which relaxes this requirement by enabling agents to use advice in a state
even if the advice is created in a slightly different state.
Compared with existing methods, agents using the proposed method have more opportunity to take advice from others.
This paper is the first to adopt the concept of differential privacy on advising to improve agent learning performance instead of addressing security issues.
The experimental results demonstrate that the proposed method is more efficient in complex environments than existing methods.
\end{abstract}

\emph{\textbf{Keywords - Multi-Agent Reinforcement Learning, Agent Advising, Differential Privacy}}

\section{Introduction}\label{sec:introduction}
Multi-agent reinforcement learning (MARL) is one of the fundamental research topics in artificial intelligence \cite{Silva19c}. 
In regular MARL methods, agents typically need a large number of interactions with the environment and other agents to learn proper behaviors.
To improve agent learning speed, the agent advising technique is introduced \cite{Silva17,Silva18,Sliva19b},
which enables agents to ask for advice between each other.

Existing advising methods share a common limitation that
an adviser agent can offer advice to an advisee agent
only if the advice is created in the same state as the advisee agent's concerned state \cite{Silva19}.
However, in complex environments, it is usually a strong requirement that two states are the same,
because a state may be composed of multiple dimensions and
two states being the same implies that all these dimensions in the two states are correspondingly identical.
This requirement may hinder the application of existing methods to complex environments.
For example, in a multi-robot search and rescue problem (Fig. \ref{fig:example}),
the aim of each robot is to search and rescue victims as quickly as possible. 
Each robot is a learning agent which can observe only its surrounding area. 
To improve robot learning performance, robots can ask for advice between each other.
In this problem, an observation of a robot is interpreted as a state of the robot. 
An observation consists of eight dimensions, 
where each dimension stands for a small cell around the robot. 
The format of an observation could be: $s=\langle 0,0,1,0,0,0,1,0\rangle$, 
where $0$ means an empty cell and $1$ means an obstacle in a cell. 
A slightly different observation could be: $s'=\langle 0,0,1,0,1,0,1,0\rangle$, 
as the two states, $s$ and $s'$, have only the fifth dimension in difference 
(referring to Definition \ref{def:difference} for more detail).
In existing advising methods, when a robot $i$ in state $s$ asks for advice from another robot $j$, 
robot $j$ can offer advice to robot $i$ 
only if robot $j$ has visited state $s$ a number of times.
\begin{figure}[ht]
\centering
	\includegraphics[scale=0.27]{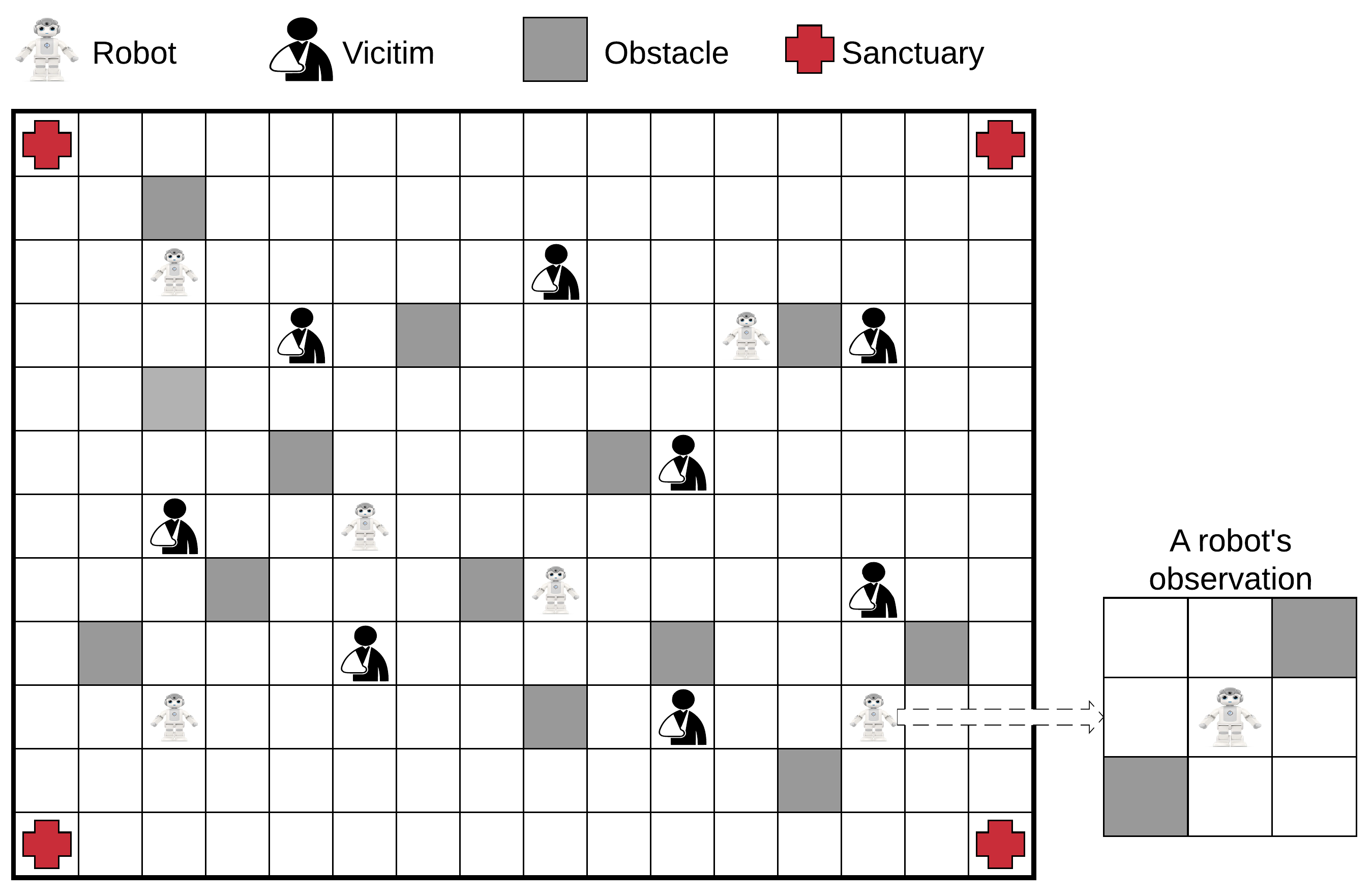}
	\caption{A multi-robot search and rescue example}
	\label{fig:example}
\end{figure}

In this paper, we relax this requirement and develop a differential advising method
which enables agents to use advice in a state
even if the advice is created in a slightly different state. 
Thus, in the above example, by using our method, even if robot $j$ has never visited state $s$, 
robot $j$ can still offer advice to robot $i$ as long as 
robot $j$ has visited a neighboring state of state $s$, e.g., state $s'$.
However, developing this differential advising method is a difficult task due to the following two challenges.
First, how to measure and delimit the difference between two states is a challenge,
because the amount of the difference between two states may significantly affect the quality of advice.
Second, how to use advice in a concerned state, which is created in another state, is also a challenge,
because improper advice may harm the learning performance of agents \cite{Ye19}.
To address these two challenges, the differential privacy technique is applied.

Differential privacy is a privacy model which guarantees that
a query to two slightly different datasets yields almost the same results \cite{Dwork06,Zhu17}.
This means that the query result yielded from one dataset can be considered approximately identical to the query result yielded from the other dataset.
This property of differential privacy can be taken into our method.
Specifically, two slightly different states are similar to two slightly different datasets.
Advices generated from states are similar to results yielded from datasets.
Since two results from two slightly different datasets can be considered approximately identical,
two advices generated from two slightly different states can also be considered approximately identical.
This property guarantees that the advice created in a state can still be used in another slightly different state.

In summary, this paper has two contributions.
First, we are the first to develop a differential advising method
which allows agents to use advices created in different states.
This method enables agents to receive more advices than existing methods.
Second, we are the first to take advantage of the differential privacy technique for agent advising 
to improve agent learning performance instead of addressing security issues.



\section{Related Work}\label{sec:related work}



Torrey and Taylor \cite{Torrey13} introduced a teacher-student framework,
where a teacher agent provides advice to a student agent to accelerate the student's learning.
Their framework is teacher-initiated,
where the teacher determines when to give advice to the student.
They developed a set of heuristic approaches for the teacher agent to decide when to provide advice.
By extending Torrey and Taylor's framework \cite{Torrey13}, Amir et al. \cite{Amir16} proposed an interactive student-teacher framework,
where the student agent and the teacher agent jointly determine when to ask for/provide advice.
They also developed a set of heuristic methods for both the student agent and the teacher agent to decide when to ask for/provide advice.

Later, Silva et al. \cite{Silva17} extended Amir et al.'s \cite{Amir16} work by taking simultaneity into consideration,
where multiple agents simultaneously learn in an environment
and an agent can be both a teacher and a student at the same time.
The decisions regarding advice request and provision are based on a set of heuristic methods.

Wang et al. \cite{Wang18b} proposed a teacher-student framework to accelerate the lexicon convergence speed among agents.
Their framework is similar to Silva et al.'s framework \cite{Silva17}.
The major difference is that, in Wang et al.'s framework,
instead of broadcasting a request to all the neighboring agents,
a student agent uses a multicast manner to probabilistically ask each neighboring agent for advice.
The asking probability is based on the distance between the student and the teacher.

Ye et al. \cite{Ye19} also proposed an agent advising approach in a simultaneous learning environment.
Unlike other advising methods which consider only benign and cooperative agents,
Ye et al. introduced malicious agents into the environment,
which provide false advice to hurt the learning performance of other agents.
Ye et al. then used the differential privacy technique to preserve the learning performance of benign agents
and reduce the impact of malicious agents. 

Omidshafiei et al. \cite{Omid19} presented an advising framework. 
Unlike existing methods which use only heuristic methods to decide when to advise, 
their framework enables agents to learn when to teach/request and what to teach. 
Specifically, an advisee agent has a student policy to decide when to request advice using advising-level observation. 
The advising-level observation is based on the advisee agentos task-level observation and action-value vectors. 
Similarly, an adviser agent has a teacher policy to decide when and what to advise using advising-level observation. 
This advising-level observation is based on 
1) the adviser agentos own task-level knowledge in the advisee agentos state, 
2) the adviser agentos measure of the advisee agentos task-level state and knowledge.

Zhu et al. \cite{Zhu19} developed a $Q$-value sharing framework named PSAF (partaker-sharer advising framework). 
Unlike other advising methods where agents transfer recommended actions as advice, 
in PSAF, agents transfer $Q$-values as advice. 
Using $Q$-values as advice is more flexible than using actions as advice. 
This is because the policies of agents may be continuously changing 
and thus an advisee agent's learning performance may be impaired by following an adviser's recommended action. 
By contrast, if the advice is a $Q$-value, the advisee can update its policy by using this $Q$-value 
and then uses its own action selection method to pick up an action.

Ilhan et al. \cite{Ilhan19} developed an action advising method in multi-agent deep reinforcement learning (MADRL). 
Unlike regular MARL, in MADRL, the number of states may be very large or even infinite. 
Thus, it is infeasible to use the number of visits to a state to measure the confidence in that state as commonly used in existing methods. 
They, hence, adopt the random network distillation to measure the confidence in a state 
by measuring the mean squared error between two neural networks: the target network and the predictor network.

Unlike the above research which focuses on the advising process, 
some works have different focuses.
Fachantidis et al. \cite{Fachantidis18} studied the critical factors that affect advice quality.
Gupta et al. \cite{Gupta19} suggested that teacher agents should not only advise the action to take in a given state
but also provide more informative advice using the synthesis of knowledge they have gained.
They proposed an advising framework where a teacher augments a student's knowledge
by providing not only the advised action but also the expected long-term reward of following that action.
Vanhee et al. \cite{Vanhee19} introduced Advice-MDPs
which extend Markov decision processes for generating policies that take into account
advising on the desirability, undesirability and prohibition of certain states and actions.

These above-mentioned works have a common limitation that
an advice can be offered only if the advice is created in the same state as the given state.
This limitation may deter the application of their works to complex environments. 
State similarity has also been researched by Castro \cite{Castro20} who defined a bisimulation metric. 
However, Castroos research focuses on computing state similarity in deterministic MDP rather than agent advising. 
Also, the definition of bisimulation provided in \cite{Castro20} may not be applicable to our research, 
because that definition requires an agent to have the full knowledge of reward functions, 
state space, and state transition functions. 
In our research, the environments are partially observable and thus, 
the full state space is unknown to any agent. 
Moreover, in our research, as multiple agents coexist in an environment, 
the state transition of an agent highly depends on the actions of other agents. 
Hence, state transition functions are hard to be pre-defined. 
In this paper, we propose a differential advising method which
overcomes the common limitation of existing methods
by taking advantage of differential privacy.

\section{Preliminaries}\label{sec:preliminaries}
\subsection{Multi-agent reinforcement learning}
Reinforcement learning is usually used to solve sequential decision-making problems.
A sequential decision-making process can be formally modeled as a Markov decision process (MDP) \cite{Watkins92}.
An MDP is typically a tuple $\langle S,A,T,R\rangle$,
where $S$ is the set of states, $A$ is a set of actions available to the agent, $T$ is the transition function, and $R$ is the reward function.

At each step, an agent observes the state $s\in S$ of the environment, and selects an action $a\in A$ based on its policy $\pi$, 
which is a probability distribution over available actions.
After performing the action, the agent receives a real-valued reward $r$ and gets into a new state $s'\in S$.
The agent then updates its policy $\pi$ based on the reward $r$ and the new state $s'$.
Through this way, the agent can gradually accumulate knowledge and improve its policy to maximize its accumulated long-term expected reward.


An MDP can be extended to model a multi-agent reinforcement learning process as a tuple $\langle S,A_{1,...,n},T,R_{1,...,n}\rangle$, where $S$ is the cartesian product of the sets of local states from each agent, $S=S_1\times S_2\cdots\times S_n$, $A_i$ is the available action set of agent $i$, $T$ is the transition function, and $R_i$ is the reward function of agent $i$. 
In multi-agent reinforcement learning, state transitions are based on the joint actions of all the agents.
Each agent's individual reward is also based on the joint actions of all the agents.

\subsection{Differential privacy}\label{sub:DP}
Differential privacy is a prevalent privacy model 
and has been broadly applied to various applications \cite{Zhu20,Ye20,Ye21}. 
Differential privacy can guarantee an individual's privacy
independent of whether the individual is in or out of a dataset \cite{Dwork06}.
Two datasets $D$ and $D'$ are neighboring datasets if they differ in at most one record.
Let $f$ be a query that maps dataset $D$ to a $k$-dimension vector in range $\mathbb{R}^k$: $f: D\rightarrow\mathbb{R}^k$.
The maximal difference on the results of query $f$ is defined as sensitivity of query $\Delta f$,
which determines how much perturbation is required for the privacy-preserving answer.
The formal definition of sensitivity is given as follows.

\begin{definition}[Sensitivity]\label{Def-GS} For a query $f:D\rightarrow\mathbb{R}^k$, the sensitivity of $f$ is defined as
\begin{equation}
\Delta f=\max\limits_{||D-D'||_1\leq 1} ||f(D)-f(D')||_1,
\end{equation}
where $||D-D'||_1\leq 1$ means that two datasets $D$ and $D'$ have at most one record different,
and $||f(D)-f(D')||_1=\sum_{1\leq i\leq k}|f(D)_i-f(D')_i|$.
\end{definition}

The aim of differential privacy is to mask the difference in the answer of query $f$ between two neighboring datasets.
To achieve this aim, differential privacy provides a randomized mechanism $\mathcal{M}$ to access a dataset.
In $\epsilon$-differential privacy, parameter $\epsilon$ is defined as the privacy budget,
which controls the privacy guarantee level of mechanism $\mathcal{M}$.
A smaller $\epsilon$ represents a stronger privacy.
The formal definition of differential privacy is presented as follows.

\begin{definition}[$\epsilon$-Differential Privacy]\label{Def-DP}
A randomized mechanism $\mathcal{M}$ gives $\epsilon$-differential privacy
if for any pair of neighboring datasets $D$ and $D'$, and for every set of outcomes $\Omega$, $\mathcal{M}$ satisfies:
\begin{equation}
Pr[\mathcal{M}(D) \in \Omega] \leq \exp(\epsilon) \cdot Pr[\mathcal{M}(D') \in \Omega]
\end{equation}
\end{definition}

One of the prevalent differential privacy mechanisms is the Laplace mechanism.
The Laplace mechanism adds Laplace noise to query results.
We use $Lap(b)$ to represent the noise sampled from the Laplace distribution with scaling $b$.
The mechanism is described as follows.

\begin{definition}[Laplace mechanism]\label{Def-LA}
Given any function $f: D \rightarrow \mathbb{R}^k$, the Laplace mechanism $\mathcal{M}_L$ is defined as
\begin{equation}\label{eq:lap}
\mathcal{M}_L(D,f,\epsilon)=f(D)+(y_1,...,y_k),
\end{equation}
where $y_1,...,y_k$ are the random noise drawn from $Lap(\frac{\Delta f}{\epsilon})$.
\end{definition}

%

\subsection{Differential advising}\label{sub:DA}
Based on the definitions of differential privacy,
we provide the definitions of differential advising. 
The aim of providing the definitions of differential advising is to 
link the properties of differential privacy to the properties of differential advising. 
This linkage can guarantee that the differential privacy mechanisms can be used as differential advising mechanisms, 
as later shown in Lemma \ref{lem:dp} and Theorem \ref{thm:DP} in Section \ref{sec:analysis}. 

The reason of applying differential privacy to advising is to ensure that 
the knowledge learned in a state can still be used in a different but similar state, i.e., a neighboring state (Definition \ref{def:difference}). 
In previous work, knowledge can only be re-used in identical states. With the introduction of the `neighboring state', a piece of knowledge has more chance to be applied by agents compared with previous work. This idea is reasonable as we human can easily borrow knowledge from a similar environment.

To achieve this goal, we can consider acquiring knowledge, i.e., asking for advice, as querying to a dataset, 
and consider the $Q$-values of state/action pairs constituting a dataset, i.e., a $Q$-table. 
Then, the differential advising process can be simulated as a differentially private query process. 
To implement differential advising, we apply the Laplace mechanism to mask the difference 
between the knowledge learned in two neighboring states. 
Specifically, a $Q$-table, shown as below, is a two-dimensional matrix 
which stores the knowledge learned by an agent. 
Differential advising is operating on $Q$-tables.

\begin{table}[!ht]\scriptsize\nonumber
\newcommand{\tabincell}[2]{\begin{tabular}{@{}#1@{}}#2\end{tabular}}
	\centering
\begin{tabular}{|c|c|c|c|c|} \hline
States/Actions & $a_1$ & $a_2$ & $a_3$ & $a_4$\\\hline
$\boldsymbol{s_1}$ & $Q(\boldsymbol{s_1},a_1)$ & $Q(\boldsymbol{s_1},a_2)$ & $Q(\boldsymbol{s_1},a_3)$ & $Q(\boldsymbol{s_1},a_4)$\\\hline
$\boldsymbol{s_2}$ & $Q(\boldsymbol{s_2},a_1)$ & $Q(\boldsymbol{s_2},a_2)$ & $Q(\boldsymbol{s_2},a_3)$ & $Q(\boldsymbol{s_2},a_4)$\\\hline
$\boldsymbol{s_3}$ & $Q(\boldsymbol{s_3},a_1)$ & $Q(\boldsymbol{s_3},a_2)$ & $Q(\boldsymbol{s_3},a_3)$ & $Q(\boldsymbol{s_3},a_4)$\\\hline
\end{tabular}
\vspace{-2mm}
	\label{tab:Q-table}
\end{table}

Without loss of generality, we assume that a state $\boldsymbol{s}$ consists of $m$ dimensions: $\boldsymbol{s}=(s_1,...,s_m)$.
Each dimension $s_i$ could be either an integer or a real number representing discrete or continuous states.
The difference between two states $\boldsymbol{s}$ and $\boldsymbol{s}'$ is defined as:
\begin{definition}[Difference between states]\label{def:difference}
\begin{equation}\label{eq:difference}
||\boldsymbol{s}-\boldsymbol{s}'||_1=\sum_{1\leq i\leq m}|s_i-s'_i|.
\end{equation}
\end{definition}
Two states $\boldsymbol{s}$ and $\boldsymbol{s}'$ are neighboring, i.e., slightly different,
if their difference is at most $1$, i.e., $||\boldsymbol{s}-\boldsymbol{s}'||_1\leq 1$. 
The difference threshold is set to $1$, 
because this is the prerequisite of using differential privacy. 
As described in Definition \ref{Def-DP}, differential privacy works 
when two datasets have at most one record in difference. 
In fact, Definition \ref{def:difference} is slightly different from the prerequisite of differential privacy. 
This is because data records in a dataset are discrete, 
and thus $||D-D'||_1\leq1$ means that one record is different in the two datasets. 
In Definition \ref{def:difference}, if states are discrete and each dimension is represented as an integer, 
$||\boldsymbol{s}-\boldsymbol{s}'||_1\leq 1$ also means that 
one dimension is different in the two states, 
such as the multi-robot example in Section \ref{sec:introduction}. 
If states are continuous, $||\boldsymbol{s}-\boldsymbol{s}'||_1\leq 1$ could mean that 
the difference exists in multiple dimensions and the sum of the difference is less than $1$.
We leave the further research of continuous states as one of our future works.

Similarly, the advice sensitivity is defined as follows.
\begin{definition}[Advice sensitivity]\label{def:as}
For an advice generation function $Q:S\rightarrow\mathbb{R}^k$,
where $S$ is the state set and $k$ is the number of actions,
the sensitivity of $Q$ is defined as
\begin{equation}
\Delta Q=\max\limits_{\boldsymbol{s},\boldsymbol{s}'\in S\wedge ||\boldsymbol{s}-\boldsymbol{s}'||_1\leq 1} ||Q(\boldsymbol{s})-Q(\boldsymbol{s}')||_{1}.
\end{equation}
\end{definition}
As an advice is a $Q$-vector, for simplicity, 
we use the letter ``$Q$'' to represent both the advice generation function and the $Q$-function. 
Details will be given in the next section.

\begin{definition}[$\epsilon$-differential advising]\label{def:da}
An advising method $\mathcal{M}_A$ is $\epsilon$-differential advising,
if for any pair of neighboring states $\boldsymbol{s}$ and $\boldsymbol{s'}$,
and for every set of advice $Ad$, $\mathcal{M}_A$ satisfies:
\begin{equation}
Pr[\mathcal{M_A}(\boldsymbol{s}) \in Ad] \leq \exp(\epsilon) \cdot Pr[\mathcal{M_A}(\boldsymbol{s'}) \in Ad]
\end{equation}
\end{definition}

The essence of differential privacy is to guarantee that 
the query results of two neighboring datasets have a very high probability to be the same. 
Similarly, in differential advising, we set that 
the query results of two neighboring $Q$-tables have a very high probability to be the same, 
so that an agent querying to two neighboring $Q$-tables can receive almost the same knowledge. 
Here, two $Q$-tables are neighboring if 1) they have one record in difference and 
2) the two different records correspond to two neighboring states. 
The query in differential advising is known as differentially private selection \cite{Dwork14} 
which produces the best answer from a space of outcomes. 
To implement differential advising, we follow the spirit of the Laplace mechanism 
by adding Laplace noise on query results to mask the difference between two neighboring $Q$-tables. 
The detail will be given in the next section.

Table \ref{tab:notation} describes the notations and terms used in this paper.
\begin{table}[!ht]\scriptsize
\newcommand{\tabincell}[2]{\begin{tabular}{@{}#1@{}}#2\end{tabular}}
	\centering
	\vspace{-5mm}
	\caption{The meaning of each notation}
\begin{tabular}{|c|c|} \hline
\textbf{notations}&\textbf{meaning}\\ \hline
$S$ & a set of states of the environment\\\hline
$A$ & a set of actions available to an agent\\\hline
$r$ & \tabincell{c}{an immediate reward obtained by \\an agent for taking an action under an observation}\\\hline
$Q(s,a)$ & \tabincell{c}{a reinforcement value for an agent\\ to take action $a$ in state $s$}\\ \hline
$\boldsymbol{\pi}$ & \tabincell{c}{a probability distribution over\\ the available actions of an agent}\\\hline
$\alpha,\zeta$ & learning rates, both of which are in $(0,1)$\\ \hline
$\gamma$ & a discount factor which is in $(0,1)$\\ \hline
$\Delta f$ & the sensitivity of a query \\\hline
$\Delta Q$ & the sensitivity of a score function \\\hline
$\epsilon$ & privacy budget \\\hline
$b$ & the scale parameter in Laplace mechanism \\\hline
$\delta$, $\beta$ & privacy budget \\\hline
$n_{\boldsymbol{s}}$ & \tabincell{c}{the number of times \\that an agent has visited a state $\boldsymbol{s}$}\\\hline
$C$ & the communication budget of an agent \\\hline
$P_{ask}$ & the probability that an agent asks for advice \\\hline
$P_{give}$ & the probability that an agent provides advice \\\hline
\end{tabular}
\vspace{-2mm}
	\label{tab:notation}
\end{table}

\section{The Differential Advising Method}\label{sec:method}
Our method is developed in a simultaneous learning framework,
where agents are learning simultaneously
and can be in both the roles of adviser and advisee.
Each agent $i$ has a communication budget $C_i$ to control its communication overhead.
Every time when an agent $i$ asks for/provides advice from/to another agent,
agent $i$'s communication budget is deducted by $1$ till the budget $C_i$ is used up.
Here, we use a combined communication budget $C_i$ for asking for and giving advice, instead of two separate budgets, 
because 1) using a combined budget is more suitable for the real-world applications than using two separate budgets, 
and 2) using a combined budget can simplify the description of our method. 
For example, in wireless sensor networks, the communication budget of a sensor is based only on its battery power 
irrespective of the communication types.

Each agent $i$ has the following knowledge:
\begin{enumerate}
  \item its available actions in each state;
  \item the $Q$-value of each available action;
  \item the existence of its neighboring agents;
\end{enumerate}
Moreover, we assume that 1) a slight change in a state will not significantly change the reward function of an agent, 
and 2) an agent in two neighboring states has the same action set.

\subsection{Overview of the method}
In our method, during advising, agents have to address the following three sub-problems:
\begin{itemize}
  \item whether to ask for advice;
  \item whether to give advice;
  \item how to use advice.
\end{itemize}
The overview of our method is outlined in Fig. \ref{fig:overview}, 
and the detail of our method is formally given in Algorithm \ref{alg:TL}.
In summary, Algorithm \ref{alg:TL} describes the workflow of our method, 
which consists of two parts: the advising part (Lines 1-15) and the learning part (Lines 16-23). 
The advising part of Algorithm \ref{alg:TL} depicts the process of an advisee agent asking for and making use of advice. 
The detail of making use of advice is given in Algorithm \ref{alg:Lap}, 
where the differential privacy mechanism is introduced. 
The learning part of Algorithm \ref{alg:TL} is a regular reinforcement learning process.
\begin{figure}[ht]
\centering
	\includegraphics[width=0.35\textwidth, height=3.6cm]{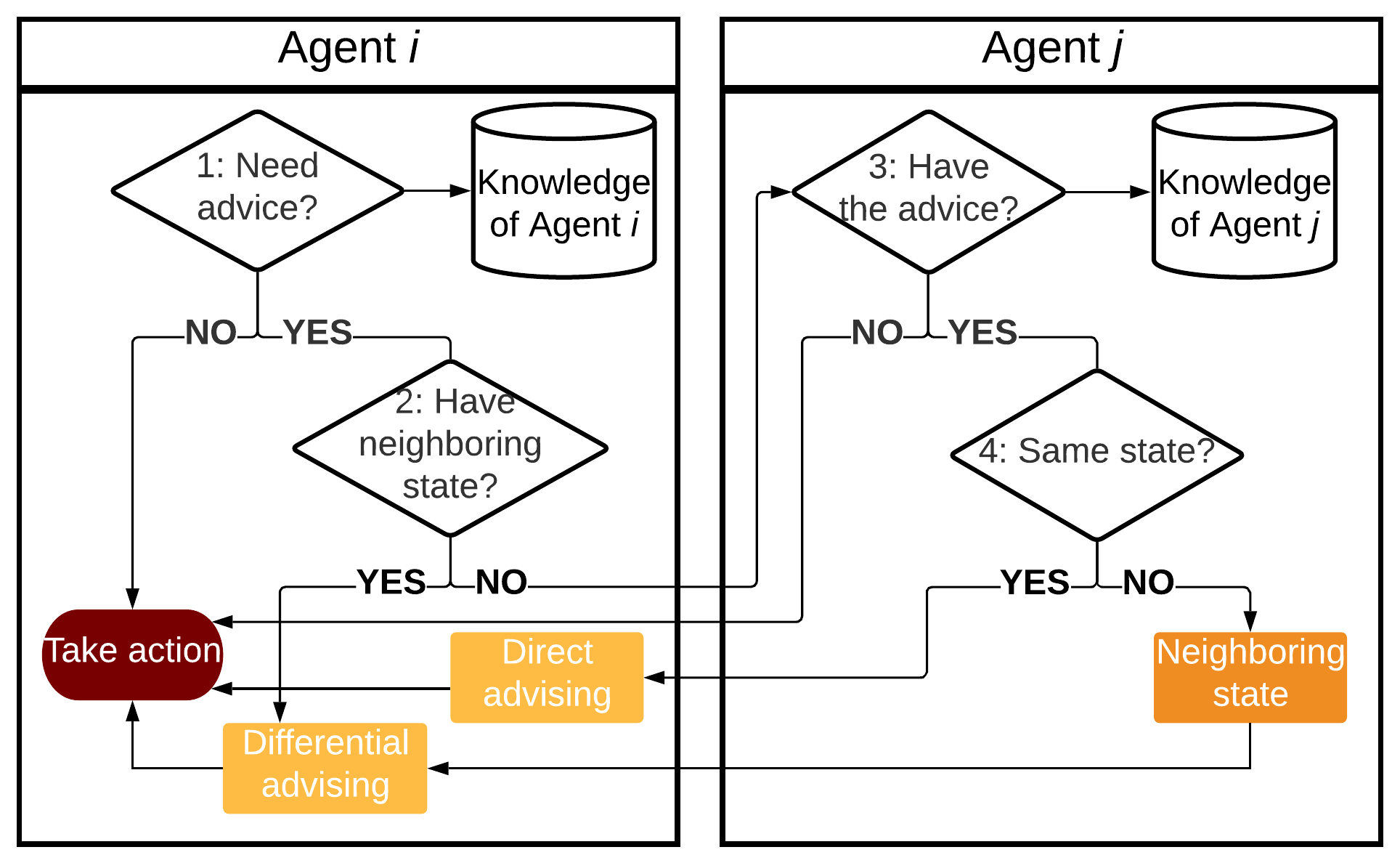}
	\caption{The overview of agent $i$ asking for advice from agent $j$}
	\label{fig:overview}
\end{figure}

\begin{algorithm}
\caption{The complete advising process}
\label{alg:TL}
/*Take agent $i$ at time step $t$ as an example */\\
Observes the environment and receives state $\boldsymbol{s}^t$;\\
Decides whether to ask for advice;\\
\If{agent $i$ decides to ask for advice}{
    \If{agent $i$ itself has a neighboring state $\boldsymbol{s}'$}{
        Uses Algorithm \ref{alg:Lap} to obtain $\boldsymbol{\pi}(\boldsymbol{s}^t)$;\\
        /* Internal advice from itself */ \\
    }\Else{
        broadcasts a message to all neighboring agents;\\
        Receives advice: $\boldsymbol{Q}(\boldsymbol{s}')$;\\
        \If{$\boldsymbol{s}'=\boldsymbol{s}^t$}{
            $\boldsymbol{\pi}(\boldsymbol{s}^t)\leftarrow\boldsymbol{\pi}(\boldsymbol{s}')$;\\
        }\Else{
            Uses Algorithm \ref{alg:Lap} to obtain $\boldsymbol{\pi}(\boldsymbol{s}^t)$;\\
            /* External advice from another agent */ \\
        }
    }
}
Selects an action $a_l$ based on the probability distribution $\boldsymbol{\pi}(\boldsymbol{s}^t)=(\pi(\boldsymbol{s}^t,a_1),...,\pi(\boldsymbol{s}^t,a_k))$;\\
$r\leftarrow$the reward for performing action $a_l$ in state $\boldsymbol{s}^t$;\\
$Q(\boldsymbol{s}^t,a_l)\leftarrow (1-\alpha)Q(\boldsymbol{s}^t,a_l)+\alpha[r+\gamma max_{a\in A_i} Q(\boldsymbol{s}^{t+1},a)]$;\\
$\overline{r}\leftarrow\sum_{a\in A_i}\pi(\boldsymbol{s}^t,a)Q(\boldsymbol{s}^t,a)$, 
where $A_i$ is the set of available actions of agent $i$ in state $\boldsymbol{s}^t$;\\
\For{each action $a\in A_i$}{
  $\pi(\boldsymbol{s}^t,a)\leftarrow\pi(\boldsymbol{s}^t,a)+\zeta(Q(\boldsymbol{s}^t,a)-\overline{r})$;\\
}
$\boldsymbol{\pi}(\boldsymbol{s}^t)\leftarrow Normalize(\boldsymbol{\pi}(\boldsymbol{s}^t))$;\\
$\alpha\leftarrow\frac{t}{t+1}\alpha$;\\
\end{algorithm}

In Lines 1-3, at time step $t$, agent $i$ observes a state $\boldsymbol{s}^t$,
and decides whether to ask for advice from other agents (Section \ref{sub:request}).
Specifically, an advice is the $Q$-value vector of a state: $\boldsymbol{Q}(\boldsymbol{s})=(Q(\boldsymbol{s},a_1),...,Q(\boldsymbol{s},a_k))$.
In Lines 4 and 5, if agent $i$ decides to ask for advice,
agent $i$ checks whether it has a neighboring state $\boldsymbol{s}'$ itself.
If so, agent $i$ takes $\boldsymbol{Q}(\boldsymbol{s}')$ as a self-advice.
Agent $i$, thereafter, uses Algorithm \ref{alg:Lap}, a differentially private algorithm, to process this advice
and selects an action based on this advice (Line 6). 
Here, allowing an agent to query itself can conserve its communication budget. 
This is called ``self-advising''.
In Lines 8 and 9, if agent $i$ has not visited any neighboring states,
agent $i$ sends a broadcast message to all communication-reachable, i.e., neighboring, agents
and its communication budget $C_i$ is deducted by $1$.
If any neighboring agent $j$ decides to offer advice,
agent $j$ sends the advice back to agent $i$
and agent $j$'s communication budget $C_j$ is deducted by $1$ (Section \ref{sub:give}).

In Lines 10 and 11, once agent $i$ receives the advice from agent $j$, 
agent $i$ checks whether this advice is created based on state $\boldsymbol{s}^t$ or a neighboring state of $\boldsymbol{s}^t$. 
If the advice is created based on $\boldsymbol{s}^t$, 
agent $i$ directly uses this advice to choose an action (Line 12). 
Otherwise, in Lines 13 and 14, agent $i$ uses Algorithm \ref{alg:Lap} to modify this advice
and then chooses an action by following the modified advice (Section \ref{sub:transfer}).
In the case that agent $i$ receives advice from multiple agents,
it selects the advice whose state has the minimum difference to agent $i$'s current state (recall Equation \ref{eq:difference}).
In the case that no advice is received or agent $i$ decides not to ask for advice,
agent $i$ simply chooses an action based on its own experience (Lines 16-23). 
Specifically, Lines 16-23 describe a regular reinforcement learning process. 
In Line 16, agent $i$ randomly selects an action based on the probability distribution $\boldsymbol{\pi}(\boldsymbol{s}^t)$ over its available actions in state $\boldsymbol{s}^t$. 
Then, in Line 17, agent $i$ performs the selected action and receives a reward $r$. 
In Line 18, agent $i$ uses reward $r$ to update the $Q$-value of the selected action. 
In Line 19, agent $i$ computes an average reward 
using the probability distribution in state $\boldsymbol{s}^t$ to multiply the $Q$-values of actions in state $\boldsymbol{s}^t$. 
After that, in Lines 20 and 21, agent $i$ adopts the average reward to update the probability distribution $\boldsymbol{\pi}(\boldsymbol{s}^t)$. 
In Line 22, $\boldsymbol{\pi}(\boldsymbol{s}^t)$ is normalized to be a valid probability distribution, 
where each element in $\boldsymbol{\pi}(\boldsymbol{s}^t)$ is in $(0,1)$ 
and the sum of all the elements is $1$. 
In Line 23, the learning rate $\alpha$ is decayed to guarantee the convergence of the algorithm.

It has to be noted that in our method, an advice is a $Q$-vector 
instead of an action suggestion or a $Q$-function. 
This is because our method is differential advising rather than direct advising, 
and using $Q$-vectors enables advisee agents to properly modify the advice before they take it. 
By contrast, modifying $Q$-vectors is easier and more reliable than modifying action suggestions and $Q$-functions.
Moreover, using $Q$-vectors as advice does not mean adviser agents must have a complete $Q$-table. 
As long as an adviser agent has the $Q$-values of the concerned state $s$ (or its neighboring state $s'$), 
the adviser can offer advice.


\subsection{Deciding whether to ask for advice}\label{sub:request}
At each time step, an agent, $i$, observes a state $\boldsymbol{s}$.
Agent $i$ decides whether to ask for advice based on probability $P_{ask}$ (Equation \ref{eq:ask}).
A random number between $0$ and $1$ is generated to compare with $P_{ask}$.
If the random number is less than $P_{ask}$,
agent $i$ asks for advice.

The calculation of $P_{ask}$ is based on 1) agent $i$'s confidence in current state $\boldsymbol{s}$,
and 2) agent $i$'s remaining communication budget, $C'_i$.
Agent $i$'s confidence in state $\boldsymbol{s}$ depends on how many times agent $i$ has visited state $\boldsymbol{s}$, denoted as $n^i_{\boldsymbol{s}}$.
A higher $n^i_{\boldsymbol{s}}$ value means a higher confidence and results in a lower probability of asking for advice.
Agent $i$'s remaining communication budget $C'_i$ depends on that how many times agent $i$ has asked for/provided advice from/to others. 
Once agent $i$ asks for advice, its communication budget is reduced by $1$: $C'_i\leftarrow C'_i-1$.
A lower remaining communication budget results in a lower probability of asking for advice. 
Formally, the probability function $P_{ask}$ takes $n^i_{\boldsymbol{s}}$ and $C'_i$ as input, 
and outputs a probability ranged in $[0,1)$.
The calculation of probability $P_{ask}$ is shown in Equation \ref{eq:ask},
where square root is used to reduce the decay speed of $P_{ask}$,
and the positive integer, $N$, is a threshold used to avoid agent $i$ using up its communication budget in very early stages.

\begin{equation}\label{eq:ask}
P_{ask}=\begin{cases}\frac{1}{\sqrt{n^i_{\boldsymbol{s}}}}\cdot\sqrt{\frac{C'_i}{C_i}}, \mbox{$n^i_{\boldsymbol{s}}\geq N$}\\
        0, \mbox{$n^i_{\boldsymbol{s}}<N$}\end{cases}
\end{equation}

If agent $i$ decides to ask for advice,
agent $i$ checks whether it has any neighboring states. 
This check is based on Equation \ref{eq:difference} to compute the difference 
between the new state $\boldsymbol{s}$ and the stored states of agent $i$. 
If the difference between $\boldsymbol{s}$ and a stored state is less than or equal to $1$, 
that stored state is a neighboring state of $\boldsymbol{s}$.

On one hand, if agent $i$ has a set of neighboring states,
agent $i$ selects the neighboring state $\boldsymbol{s}'$,
which has the minimum difference to its current state $\boldsymbol{s}$.
Agent $i$ utilizes the proposed differential advising method
to modify $\boldsymbol{Q}(\boldsymbol{s}')$ and
use the modified $\boldsymbol{Q}(\boldsymbol{s}')$ to take an action (Section \ref{sub:transfer}).
On the other hand, if agent $i$ does not have any neighboring states,
it asks for advice from its neighboring agents (Section \ref{sub:give}).

\subsection{Deciding whether to give advice}\label{sub:give}
When an agent, $j$, receives an advice request message from agent $i$,
agent $j$ needs to decide whether to provide advice to agent $i$.
This decision is based on probability $P_{give}$ (Equation \ref{eq:give}).
The calculation of $P_{give}$ is based on 1) agent $j$'s confidence in agent $i$'s concerned state $\boldsymbol{s}$ or neighboring states of $\boldsymbol{s}$,
and 2) agent $j$'s remaining communication budget $C'_j$.
If agent $j$ has more than one neighboring state,
agent $j$ uses the one which has the minimum difference to state $\boldsymbol{s}$.
The higher the confidence of agent $j$, the higher the probability agent $j$ will provide advice.
The lower remaining communication budget of agent $j$, the lower probability agent $j$ will provide advice. 

Formally, the probability function $P_{give}$ takes $n^j_{\boldsymbol{s}}$ and $C'_j$ as input, 
and outputs a probability ranged in $[0,1)$. 
Here, for simplicity, we use $n^j_{\boldsymbol{s}}$ to denote the number of times that 
agent $j$ has visited either state $\boldsymbol{s}$ or a neighboring state of $\boldsymbol{s}$, 
and $C'_j$ to denote the remaining communication budget of agent $j$. 
Once agent $j$ gives advice, its remaining communication budget is reduced by $1$: $C'_j\leftarrow C'_j-1$.
It should be noted that agent $j$ does not need to know the full state space to find neighboring states, 
because in some environments, it is infeasible for any individual agent to know the full state space, 
such as the multi-robot example in Section \ref{sec:introduction}. 
Thus, agent $j$ simply checks the states which it has visited. 
If all the visited states are not neighboring to the concerned state, 
this means that agent $j$ does not have the experience related to the concerned state 
and should not offer any advice to agent $i$.
Formally, the calculation of $P_{give}$ is given in Equation \ref{eq:give}.
\begin{equation}\label{eq:give}
P_{give}=\begin{cases}(1-\frac{1}{\sqrt{n^j_{\boldsymbol{s}}}})\cdot\sqrt{\frac{C'_j}{C_j}}, \mbox{$n^j_{\boldsymbol{s}}\geq n^i_{\boldsymbol{s}}$}\\
        0, \mbox{$n^j_{\boldsymbol{s}}<n^i_{\boldsymbol{s}}$}\end{cases}
\end{equation}

In Equation \ref{eq:give}, we set a thresholds $n^i_{\boldsymbol{s}}$ for adviser agents. 
If an adviser agent has visited a state for less than $n^i_{\boldsymbol{s}}$ times, 
that means the adviser agent is less experienced than the advisee agent in state $\boldsymbol{s}$ 
and thus the adviser agent is not allowed to offer advice to the advisee agent for that state. 
This setting is based on the fact that in theory, 
an adviser agent can generate advice for any possible states, 
even if the adviser agent has never visited those states. 
For example, an adviser agent can generate advice for a never visited state 
using the initialized $Q$-values of actions in that state. 
However, since the $Q$-value of each action can be randomly initialized, 
such advice is useless and may even harm the learning performance of the advisee.

\subsection{How to use advice}\label{sub:transfer}
When agent $i$ receives agent $j$'s advice, $\boldsymbol{Q}(\boldsymbol{s}')$, 
agent $i$ checks whether its concerned state $\boldsymbol{s}$ is the same as state $\boldsymbol{s}'$. 
The aim of this check is to decide whether this advice is created based on the same state ($\boldsymbol{s}=\boldsymbol{s}'$) 
or a neighboring state ($\boldsymbol{s}\neq \boldsymbol{s}'$).
If $\boldsymbol{s}=\boldsymbol{s}'$, agent $i$ directly takes agent $j$'s advice and selects an action based on agent $j$'s advice.
If $\boldsymbol{s}\neq \boldsymbol{s}'$, 
agent $i$ uses the Laplace mechanism to add noise to agent $j$'s advice, $\boldsymbol{Q}(\boldsymbol{s}')$, shown in Algorithm \ref{alg:Lap}.
Here is an example. 
Agent $i$ arrives at state s and asks for advice from agent $j$. 
Agent $j$ has visited three states: $\boldsymbol{s}_0,\boldsymbol{s}_1,\boldsymbol{s}_2$ and are confident in these states. 
After calculation, agent $j$ detects that none of the three states is the same as state $\boldsymbol{s}$ 
and only $\boldsymbol{s}_0$ is a neighboring state of $\boldsymbol{s}$. 
Thus, agent $j$ sends its advice, Q-vector $Q(\boldsymbol{s}_0)$, to agent $i$. 
After agent $i$ receives this advice, 
agent $i$ identifies that this advice is created based on a neighboring state $\boldsymbol{s}_0$. 
Agent $i$ adopts Algorithm \ref{alg:Lap} to add Laplace noise to the advice $Q(\boldsymbol{s}_0)$, 
and then uses the modified advice to guide its action selection.

\begin{algorithm}
\caption{Differential advising}
\label{alg:Lap}
\textbf{Input}: $\boldsymbol{Q}(\boldsymbol{s}')=(Q(\boldsymbol{s}',a_1),...,Q(\boldsymbol{s}',a_k))$;\\
   \For{each action $a\in A_i$}{
       adjusting $Q$-value $Q(\boldsymbol{s}', a)$ as:
       $Q(\boldsymbol{s}',a)\leftarrow Q(\boldsymbol{s}',a)+Lap(\frac{\Delta Q}{\epsilon})$;\\
   }
$\overline{r}\leftarrow\sum_{a\in A_i}\pi(\boldsymbol{s}^t,a)Q(\boldsymbol{s}',a)$;\\
\For{each action $a\in A_i$}{
  $\pi(\boldsymbol{s}^t,a)\leftarrow\pi(\boldsymbol{s}^t,a)+\zeta(Q(\boldsymbol{s}',a)-\overline{r})$;\\
}
$\boldsymbol{\pi}(\boldsymbol{s}^t)\leftarrow Normalize(\boldsymbol{\pi}(\boldsymbol{s}^t))$;\\
\textbf{Output}: $\pi(\boldsymbol{s}^t,a_1),...,\pi(\boldsymbol{s}^t,a_k)$;\\
\end{algorithm}

In Lines 2 and 3, agent $j$'s $Q$-values are adjusted by adding Laplace noise.
Specifically, the Laplace noise is added to each of the $Q$-values: $Q(\boldsymbol{s'},a_1)+Lap_1,m,Q(\boldsymbol{s'},a_k)+Lap_k$, 
where each noise, $Lap_i$, is a random number based on the Laplace distribution.
Based on the adjusted $Q$-values, in Lines 4-6, agent $i$ computes its average reward in state $\boldsymbol{s}^t$ 
and updates its probability distribution $\pi(\boldsymbol{s}^t,a)$. 
Note that the computation of the average reward uses the probability distribution in state $\boldsymbol{s}^t$ to multiply the $Q$-values of the available actions in state $\boldsymbol{s'}$. 
As we have the assumption that an agent has the same action set in two neighboring states, 
this multiplication is valid.
In Line 7, $\pi(\boldsymbol{s}^t,a)$ is normalized to be a valid probability distribution.

In the Laplace mechanism in Line 3, $\Delta Q$ is the sensitivity of $Q$-values (recall Definition \ref{def:as}).
Typically, $\Delta Q$ is based on $Q$-functions, and $Q$-functions are based on rewards. 
Thus, essentially, $\Delta Q$ is based on rewards, and we can use rewards to approximately compute $\Delta Q$. 
We have to notice that it is infeasible to accurately compute $\Delta Q$ before the learning starts, 
because the accurate computation of $\Delta Q$ needs the final $Q$-values of all the states, 
and these final $Q$-values can only be obtained when the learning finishes. 

We use an example to demonstrate how to approximately compute $\Delta Q$. 
Definition \ref{def:as} states that $\Delta Q$ is the maximum difference between two $Q$-vectors of any two neighboring states. 
In the multi-robot problem exemplified in Section \ref{sec:introduction}, for any two neighboring states, 
the maximum reward difference is $10-(-5)=15$ 
given that the maximum reward $10$ is finding a victim 
and the minimum reward $-5$ is hitting an obstacle. 
This is because in any two neighboring states, only one dimension, i.e., one cell, is different. 
This different cell incurs a maximum reward difference, 
if the cell has a victim in a state 
while has an obstacle in the other state. 
Moreover, $Q$-value is usually updated based on a learning rate $\alpha$. 
Thus, in the multi-robot problem, $\Delta Q$ can be approximately computed as $\alpha\cdot 15$. When we set $\alpha=0.2$, then $\Delta Q=3$.

The rationale of Algorithm \ref{alg:Lap} is discussed as follows.
According to Definition \ref{Def-LA} in Section \ref{sub:DP}, 
the Laplace mechanism is used to add Laplace noise to the numerical output of a query $f$ to a dataset $D$, 
so that both the privacy and the utility of the output can be guaranteed. 
Typically, a query $f$ maps a dataset $D$ to $k$ real numbers: $f: D \rightarrow \mathbb{R}^k$. 
In comparison, in advising learning, the $Q$-table of an adviser agent is interpreted as a dataset, 
where each entry in the $Q$-table is the $Q$-value of a pair of a state and an action, $Q(s,a)$. 
Advice request from an advisee agent is interpreted as the query to the Q-table of the adviser agent. 
The output of the query is the advice of the adviser agent, 
which consists of $k$ real numbers, $Q(\boldsymbol{s},a_1),...,Q(\boldsymbol{s},a_k)$, 
and $k$ is the number of actions in state $s$. 
As discussed in Section \ref{sub:DA}, we have connected the properties of differential privacy and the properties of differential advising. 
Hence, by adding Laplace noise to $Q(\boldsymbol{s},a_1),...,Q(\boldsymbol{s},a_k)$, 
both the validity and utility of the advice can be guaranteed, as shown in Theorem \ref{thm:DP}, \ref{thmUtility} and \ref{thm:utility2} in Section \ref{sec:analysis}.
Thus, in state $\boldsymbol{s}$, agent $i$ can still use the advice $\boldsymbol{Q}(\boldsymbol{s}')$,
which is created based on a neighboring state $\boldsymbol{s}'$.


\section{Theoretical Analysis}\label{sec:analysis}

\subsection{Differential advising analysis}
\begin{lem}\label{lem:dp}
Any methods, which satisfy $\epsilon$-differential privacy and
meet the input and output requirements of $\mathcal{M}_A$ in Definition \ref{def:da},
also satisfy $\epsilon$-differential advising.
\end{lem}
\begin{proof}
In Definition \ref{Def-DP}, the input of $\mathcal{M}$ can be any type of dataset,
while the output can be any valid results.
In Definition \ref{def:da}, the input of $\mathcal{M}_A$ is a state $\boldsymbol{s}$ which is a vector,
while the output is also a vector $\boldsymbol{Q}(\boldsymbol{s})$.
Thus, Definition \ref{def:da} can be considered a special case of Definition \ref{Def-DP}.
Formally, let $D\in\mathbb{N}^{|\mathcal{X}|}$ and $\boldsymbol{s}\in S$,
where $\mathbb{N}^{|\mathcal{X}|}$ is a data universe with $|\mathcal{X}|$ attributes and $S$ is a state space.
As $\mathbb{N}^{|\mathcal{X}|}$ is a data universe and can be arbitrarily large, 
we have $S\subset\mathbb{N}^{|\mathcal{X}|}$.
Similarly, we also have $Ad\subset\Omega$.
Therefore, any methods, which satisfy $\epsilon$-differential privacy and
meet the input and output requirements of $\mathcal{M}_A$,
also satisfy $\epsilon$-differential advising.
\end{proof}

Lemma \ref{lem:dp} creates a link between differential privacy and differential advising. 
Then, differential privacy mechanisms, e.g., the Laplace mechanism, can be used to implement differential advising.

\begin{thm}\label{thm:DP}
The proposed method is $\epsilon$-differential advising.
\end{thm}
\begin{proof}
According to Lemma \ref{lem:dp}, to prove this theorem,
we need only to prove that the proposed method satisfies $\epsilon$-differential privacy.
In the proposed method, Algorithm \ref{alg:Lap} utilizes a differential privacy mechanism, i.e., the Laplace mechanism, in Line 3.
We re-write Line 3 in the $\boldsymbol{Q}$-vector form as $\mathcal{M}_a(\boldsymbol{s},\boldsymbol{Q},\epsilon)=\boldsymbol{Q}(\boldsymbol{s})+(Lap_1,...,Lap_k)$,
where $Lap_1,...,Lap_k$ are random noise sampled from $Lap(\frac{\Delta Q}{\epsilon})$
and $k$ is the number of actions in state $\boldsymbol{s}$.
Since the Laplace mechanism is differentially private,
by comparing Equation \ref{eq:lap} with the re-written equation of Line 3,
we can conclude that $\mathcal{M}_a(\boldsymbol{s},\boldsymbol{Q},\epsilon)$ satisfies $\epsilon$-differential privacy.
\end{proof}

Theorem \ref{thm:DP} theoretically proves that 
the proposed method is a valid differential advising method 
which enables agents to use advice created based on similar states. 
Moreover, the parameter $\epsilon$ is used to control the amount of noise added to an advice. 
The tuning of $\epsilon$ is left to users.

\subsection{Average reward analysis}
As shown in Line 4, Algorithm \ref{alg:Lap}, the average reward of an agent is:
$\overline{r}=\sum_{a\in A}\pi(\boldsymbol{s},a)Q(\boldsymbol{s},a)$.
We will demonstrate that with a very high probability,
the average reward of an agent using differential advising is greater than the average reward without using differential advising.

\begin{definition}[($\delta$,$\beta$)-useful]\label{DEF-USE}
A differential advising multi-agent system is ($\delta$,$\beta$)-useful
if for each agent, we have $Pr(\overline{r}'-\overline{r}>\delta)<1-\beta$,
where $\overline{r}'$ and $\overline{r}$ are the average rewards of an agent using and without using differential advising, respectively, and $\delta>0$ and $0<\beta<1$.
\end{definition}

\begin{thm}\label{thmUtility}
For any $0<\beta<1$, with probability at least $1-\beta$,
the average reward of an agent using differential advising
is greater than the average reward without using it. 
The difference is bounded by $\delta$,
which satisfies $0<\delta<-\frac{\Delta Q}{\epsilon}\cdot\frac{1}{t}\cdot ln(2-2\beta)$, 
where $t$ is the number of learning iterations.
\end{thm}

\begin{proof}
The average reward of an agent is $\overline{r}=\sum_{a\in A}\pi(\boldsymbol{s},a)Q(\boldsymbol{s},a)$.
After applying the Laplace mechanism,
the average reward becomes $\overline{r}'=\sum_{a\in A}\pi(\boldsymbol{s},a)(Q(\boldsymbol{s},a)+Lap(b))$.
Now, we need to prove that $Pr(\overline{r}'-\overline{r}>\delta)>1-\beta$.

According to calculation, we have 
$\overline{r}'-\overline{r}=\sum_{a\in A}\pi(\boldsymbol{s},a)(Q(\boldsymbol{s},a)+Lap_a(b))-\sum_{a\in A}\pi(\boldsymbol{s},a)Q(\boldsymbol{s},a)=\sum_{a\in A}\pi(\boldsymbol{s},a)Lap_a(b)$.
Thus, $Pr(\overline{r}'-\overline{r}>\delta)=Pr(\sum_{a\in A}\pi(\boldsymbol{s},a)Lap_a(b)>\delta)$. 
Because $\sum_{a\in A}\pi(\boldsymbol{s},a)=1$, we have 
$Pr(\sum_{a\in A}\pi(\boldsymbol{s},a)Lap_a(b)>\delta)=Pr(\sum_{a\in A}\pi(\boldsymbol{s},a)Lap_a(b)>\sum_{a\in A}\pi(\boldsymbol{s},a)\delta)$.
The property of $Lap(b)$ is presented in Lemma~\ref{lem1}.
\begin{lem}\label{lem1}
(Laplace Random Variables).
For every $\delta>0$,
\begin{equation}
Pr(Lap(b)>\delta)=\frac{1}{2}\exp(-\frac{\delta}{b}).
\end{equation}
\end{lem}
\begin{proof}
The proof of this lemma can be found in \cite{Kasiviswanathan2008531}.
\end{proof}
Based on this property, we have 
$Pr(\overline{r}'-\overline{r}>\delta)=Pr(\sum_{a\in A}\pi(\boldsymbol{s},a)Lap_a(b)>\sum_{a\in A}\pi(\boldsymbol{s},a)\delta)=\frac{1}{2}\exp(-\frac{\delta}{b})$.

Lemma \ref{lem1} demonstrates the probability in one learning iteration.
If an agent learns $t$ iterations,
the probability becomes $Pr(\overline{r}'-\overline{r}>\delta)=[\frac{1}{2}\exp(-\frac{\delta}{b})]^t$.
To guarantee $Pr(\overline{r}'-\overline{r}>\delta)>1-\beta$,
we have $[\frac{1}{2}\exp(-\frac{\delta}{b})]^t>1-\beta$.
After calculation, we have $\delta<-b\cdot\frac{1}{t}\cdot ln(2-2\beta)$.
As $b=\frac{\Delta S}{\epsilon}$, we have $\delta<-\frac{\Delta Q}{\epsilon}\cdot\frac{1}{t}\cdot ln(2-2\beta)$.
\end{proof}

In this proof, we use the average reward equation: $\overline{r}=\sum_{a\in A}\pi(\boldsymbol{s},a)Q(\boldsymbol{s},a)$. 
This equation is valid to correctly estimate an agent's average reward, 
only if the agent has a good estimation of the $Q$-table. 
Hence, this theorem holds in late learning stages 
and reveals the final results.

\textbf{Remark 1}: The upper bound of $\delta$, $-\frac{\Delta Q}{\epsilon}\cdot\frac{1}{t}\cdot ln(2-2\beta)$, relies on $t$.
Let $\delta(t)=-\frac{\Delta Q}{\epsilon}\cdot\frac{1}{t}\cdot ln(2-2\beta)$.
Hence, $\delta(t)$ is monotonically decreasing with the increase of $t$.
Specifically, when $t\rightarrow 0$, $\delta(t)\rightarrow\infty$,
and when $t\rightarrow\infty$, $\delta(t)\rightarrow 0$.
This means that in early stages, using differential advising can significantly increase agents' average rewards.
However, as time progresses, the improvement of agents' average rewards decreases.
This situation reflects the fact that in early stages, an agent is not knowledgable,
so other agents' knowledge can give the agent significant help.
As time progresses, the agent accumulates enough knowledge,
and thus other agents' knowledge is trivial to the agent.

Theorem \ref{thmUtility} theoretically demonstrates the effectiveness of our method 
by analyzing lower bound of an agent's average reward difference between using and not using our method. 
The analysis result shows that the lower bound of the average reward difference is positive with a very high probability. 
This means that the agent using our method can receive a higher average reward than not using our method with a very high probability. 
Given that the average reward of using and not using the proposed method is $\overline{r}'$ and $\overline{r}$, respectively, 
the parameter $\beta$ is used to measure the lower bound of the average reward difference, $\overline{r}'-\overline{r}$. 
Particularly, when $\beta\rightarrow 0$, the lower bound is a small positive number; 
and when $\beta\rightarrow 1$, the lower bound tends to be $+\infty$. 
This means that the average reward difference, $\overline{r}'-\overline{r}$, 
has a high probability, $1-\beta$, to be a small positive number, 
while has a low probability to be very large. 
Certainly, the average reward difference, $\overline{r}'-\overline{r}$, cannot be $+\infty$. 
Thus, for the completeness of Theorem \ref{thmUtility}, as a supplement, 
Theorem \ref{thm:utility2} gives the upper bound of this difference, $\overline{r}'-\overline{r}$.

\begin{thm}\label{thm:utility2}
Let $k$ be the number of actions in state $\boldsymbol{s}$. 
Let $v\geq\sqrt{k}\frac{\Delta Q}{\epsilon}$ and $0<\lambda<\frac{2\sqrt{2}\epsilon v^2}{\Delta Q}$. 
Then $Pr(\overline{r}'-\overline{r}>\lambda)\leq exp(-\frac{t\lambda^2}{8v^2})$, 
where $t$ is the number of learning iterations.
\end{thm}
\begin{proof}
As shown in Theorem \ref{thmUtility}, we have 
$\overline{r}'-\overline{r}=\sum_{a\in A}\pi(\boldsymbol{s},a)(Q(\boldsymbol{s},a)+Lap(b))-\sum_{a\in A}\pi(\boldsymbol{s},a)Q(\boldsymbol{s},a)=\sum_{a\in A}\pi(\boldsymbol{s},a)Lap_a(b)$. 
Since each $\pi(\boldsymbol{s},a)$ is in $[0,1]$, we have 
$\overline{r}'-\overline{r}=\sum_{a\in A}\pi(\boldsymbol{s},a)Lap_a(b)<\sum_{a\in A}Lap_a(b)$.
Thus, we have $Pr(\overline{r}'-\overline{r}>\lambda)<Pr(\sum_{a\in A}Lap_a(b)>\lambda)$.
The property of $\sum_{a\in A}Lap_a(b)$ is given in Lemma \ref{lem2}.
\begin{lem}\label{lem2}
(Sum of Laplace random variables) 
Let $Lap_1(b),...,Lap_k(b)$ be $k$ independent variables with distribution $Lap(b)$, 
then $Pr(\sum_{a\in A}Lap_a(b)>\lambda)\leq exp(-\frac{\lambda^2}{8v^2})$.
\end{lem}
\begin{proof}
The proof of this lemma can be found in \cite{Dwork14}.
\end{proof}
Lemma \ref{lem2} demonstrates the probability in one learning iteration.
If an agent learns $t$ iterations,
the probability becomes $Pr(\overline{r}'-\overline{r}>\lambda)\leq exp(-\frac{t\lambda^2}{8v^2})$.
\end{proof}

According to Theorem \ref{thmUtility} and \ref{thm:utility2}, 
it can be found that the performance of the proposed method is affected by the advice sensitivity $\Delta Q$, 
and a smaller $\Delta Q$ means a better performance. 
Since advice sensitivity, $\Delta Q$, is based on the $Q$-values of states, 
we can conclude that the proposed method can work well in the environments 
where the $Q$-values of any pair of neighboring states have a small difference.

\subsection{Convergence analysis}
In Algorithm \ref{alg:TL}, let $\alpha_t$ be the learning rate at time step $t$.
Since $\alpha_{t+1}=\frac{t}{t+1}\cdot\alpha_t=\frac{t}{t+1}\cdot\frac{t-1}{t}\cdot\alpha_{t-1}=...$,
we can conclude that $\alpha_t=\frac{1}{t}\alpha_1$.
Let $t_i$ be the index of the $i$th time that action $a$ is performed in state $\boldsymbol{s}$.
For example, suppose that the $1$st time that action $a$ is performed in state $\boldsymbol{s}$ is at time step $6$, then $t_1=6$.
To prove the convergence of Algorithm \ref{alg:TL},
we need the results of Lemma \ref{thm:theorem1} \cite{Anton12} and Lemma \ref{thm:theorem2} \cite{Watkins92}.

\begin{lem}\label{thm:theorem1}
A series, $s_1+s_2+...$, converges if and only if for any given small positive number $\mu>0$,
there is always an integer $N$, such that when $n>N$, for any positive integer $m$,
the inequality holds: $|s_{n+1}+...+s_{n+m}|<\mu$.
\end{lem}

It should be noted that in Lemma \ref{thm:theorem1}, if $|s_{n+1}+...+s_{n+m}|\geq\mu$,
the series, $s_1+s_2+...$, does not converge.

\begin{lem}\label{thm:theorem2}
Given bounded rewards and learning rate $0\leq\alpha_{t_i}\leq 1$,
for any pair of state $\boldsymbol{s}$ and action $a$, if series $\sum_{1\leq i}\alpha^2_{t_i}$ is bounded, i.e., $\sum_{1\leq i}\alpha^2_{t_i}<+\infty$,
and series $\sum_{1\leq i}\alpha_{t_i}$ is unbounded, i.e., $\sum_{1\leq i}\alpha_{t_i}\rightarrow+\infty$,
then, when $i\rightarrow+\infty$, $Q(\boldsymbol{s},a)$ converges, with probability $1$, to the optimal $Q$-value, i.e., the expected reward.
\end{lem}

\begin{thm}\label{thm:convergence1}
In Algorithm \ref{alg:TL}, for any pair of state $\boldsymbol{s}$ and action $a$,
as $i\rightarrow+\infty$, $Q(\boldsymbol{s},a)$ converges, with probability $1$, 
to the expected reward of an agent performing action $a$ in state $\boldsymbol{s}$.
\end{thm}
\begin{proof}
Based on Lemma \ref{thm:theorem2}, to prove the convergence,
we need only to prove that in Algorithm \ref{alg:TL}, $\sum_{1\leq i}\alpha^2_{t_i}<+\infty$ and $\sum_{1\leq i}\alpha_{t_i}\rightarrow+\infty$ always hold
for any pair of state $\boldsymbol{s}$ and action $a$.

First, we prove that $\sum_{1\leq i}\alpha^2_{t_i}<+\infty$ always holds.
It has been known that $\sum_{1\leq t}\frac{1}{t^2}=\frac{\pi^2}{6}$, where $t$ means time step, a positive integer.
Then, the following inequality can be concluded:
$\sum_{1\leq i}\alpha^2_{t_i}=\sum_{1\leq i}\frac{\alpha^2_1}{t^2_i}\leq\sum_{1\leq t}\frac{\alpha^2_1}{t^2}=\alpha^2_1\frac{\pi^2}{6}<+\infty$.

Then, we prove that $\sum_{1\leq i}\alpha_{t_i}\rightarrow+\infty$ always holds.
Given a small positive number $\mu$ and a positive integer $N$,
let $0<\mu<\beta\cdot prob(\boldsymbol{s})\cdot\frac{\alpha_1}{2}$.
Here, $\beta$ is the mapping lower bound, a small positive number, used to normalize probability distribution $\boldsymbol{\pi}(\boldsymbol{s})$,
and $prob(\boldsymbol{s})$ is the probability with which an agent observes state $\boldsymbol{s}$.
As $\pi(\boldsymbol{s},a)\leq\beta>0$, there is at least one integer $n>N$,
such that in $(n+1)$th, $(n+2)$th, ..., $(n+n)$th time steps,
the number of times that action $a$ is performed in state $\boldsymbol{s}$ is greater than or equal to $\beta\cdot prob(\boldsymbol{s})\cdot n$.
Therefore, $\sum_{n<t_i\leq 2n}\alpha_{t_i}=\sum_{n<t_i\leq 2n}\frac{\alpha_1}{t_i}\geq\sum_{n<t_i\leq 2n}\frac{\alpha_1}{2n}\geq\beta\cdot prob(\boldsymbol{s})\cdot n\cdot\frac{\alpha_1}{2n}=\beta\cdot prob(\boldsymbol{s})\cdot\frac{\alpha_1}{2}>\mu$.
Based on Lemma \ref{thm:theorem1}, as $\sum_{n<t_i\leq 2n}\alpha_{t_i}>\mu$,
the series, $\sum_{i}\alpha_{t_i}$, does not converge,
which means that $\sum_{i}\alpha_{t_i}\rightarrow+\infty$.

Because $\sum_{1\leq i}\alpha^2_{t_i}<+\infty$ and $\sum_{i}\alpha_{t_i}\rightarrow+\infty$,
based on Lemma \ref{thm:theorem2}, it can be concluded that
in Algorithm \ref{alg:TL}, when $i\rightarrow+\infty$, $Q(\boldsymbol{s},a)$ converges, with probability $1$, to the expected reward of a player.
\end{proof}

Theorem \ref{thm:convergence1} theoretically proves the convergence of our method, 
which means that agents using our method can finally receive their expected rewards.



\section{Experiments}\label{sec:experiment}

\subsection{Experimental setup}
Two experimental scenarios are used to evaluate the proposed differential advising method, denoted as \emph{DA-RL} (differential advising reinforcement learning).
\subsubsection{Scenario 1}
The first scenario is the multi-robot problem \cite{Wang08} shown in Fig. \ref{fig:example}.
Each agent represents a robot which has four actions: \emph{up}, \emph{down}, \emph{left} and \emph{right}. 
An observation of an agent is interpreted as its state which consists of $8$ dimensions.
The aim of the agents is to achieve the targets on the map,
which could be victims in search and rescue or rubbish in office cleaning.
The reward for achieving a target is set to $10$.
The reward for hitting an obstacle is set to $-5$.
The reward for moving a step is $0$. 
The reward in each situation is pre-defined as knowledge to each agent. 
For example, when an agent achieves a target, its reward is automatically added by $10$, 
while when the agent hits an obstacle, its reward is automatically reduced by $5$. 
Moreover, agents cannot be at the same cell at the same time step.

The first scenario is classified as three settings.
Setting 1 (\emph{Static}): agents achieve fixed targets.
Setting 2 (\emph{Dynamic 1}): agents achieve dynamically generated targets.
Setting 3 (\emph{Dynamic 2}): agents achieve dynamically moving targets.
		
The size of environments varies from $12\times 8$ cells to $24\times 16$ cells.
The number of agents varies from $2$ to $6$.
The number of targets varies from $20$ to $60$.
The number of obstacles varies from $15$ to $35$.

In this scenario, three evaluation metrics are used.

Metric 1 (Average time): the number of time steps to achieve all targets in average in one learning round.
    This metric is specific for the first and third settings: \emph{Static} and \emph{Dynamic 2}.
    Here, a learning round is a period during which the agents achieve all the targets.
		
Metric 2 (Average time for one target): the number of time steps to achieve one target in average in one learning round.
    This metric is specific for the second setting: \emph{Dynamic 1}. 
		The reason of using different metrics in different settings will be given at the end of this section.
		
Metric 3 (Average hit): the number of hits to obstacles in average in one learning round.

\subsubsection{Scenario 2}
This scenario is the multi-agent load balancing \cite{Sakuma08}. 
Each agent represents a factory.
Each factory has a backlog which stores a set of types of items.
Let the number of item types be $k$ and the maximum stock for each type $i$ be $M_i$.
A state of an agent is the current number of items of each type in stock.
A state consists of $k$ dimensions.

At each time step, a random item from each agent is processed with a given probability $p_s$,
and a new item arrives at each agent with another given probability $p_a$.
The actions of an agent include whether or not to pass an arrived item to another agent.
The reward of an agent $A$ is based on its backlog: $r_A=50-\sum_{1\leq i\leq k}w_i\cdot m_i$,
where $w_i$ is the importance weight of items of type $i$ 
and $m_i$ is the current number of items of type $i$.
In addition, if agent $A$ passes an item $i$ to agent $B$,
agent $A$'s reward is reduced by $2\cdot w_i$ as a cost for redistribution.
The aim of agents is to maximize their accumulated reward.

The number of agents varies from $2$ to $6$.
Item types varies from $2$ to $6$, 
where for two types, the weights are set to $5$ and $4$; 
for three types, the weights are set to $5$, $4$ and $3$; and so on.
The maximum stock for each type of item varies from $3$ to $7$.
The probability of processing an item $p_s$ varies from $0.2$ to $0.6$.
The probability of arriving at an item $p_a$ is set to $0.4$.

The evaluation metric is the average reward of all the agents in one learning round. 
Here, a learning round means a pre-defined number of learning iterations.

\subsubsection{The methods for comparison}
In this experiment, two methods are involved for comparison.
The first method is a regular reinforcement learning method, denoted as \emph{RL},
which does not include any advising.
The \emph{RL} method is used as a comparison standard.
The second method is from \cite{Silva17}, denoted as \emph{SA-RL} (simultaneous advising reinforcement learning). 
The \emph{SA-RL} method is an Ad hoc TD advising method, 
where agents simultaneously and independently learn and share advice when needed.
The \emph{SA-RL} method is similar to our method.
The difference includes 1) the \emph{SA-RL} method allows agents to offer advice only in the same states;  
and 2) agents in the \emph{SA-RL} method transfers recommended actions as advice. 
Moreover, the \emph{SA-RL} method involves two parameters, $v_a$ and $v_g$, used to control the probability of asking for and giving advice, respectively. 
A higher $v_a$ value results in a lower probability of asking for advice, 
while a higher $v_g$ results in a higher probability of giving advice.
The \emph{SA-RL} method is used to demonstrate the effectiveness of the differential privacy technique on agent advising.

The values of the parameters used in the experiments are set as follows:
$\alpha=0.2$, $\gamma=0.8$, $\zeta=0.1$, $\epsilon=1$, $v_a=0.4$ and $v_g=0.9$.
The values of these parameters are experimentally chosen to yield good results. 
Moreover, $\Delta Q=3$ in Scenario 1, while $\Delta Q=1$ in Scenario 2. 
The experimental results are obtained by averaging the results of $500$ runs.

\subsection{Experimental results of scenario 1}
\subsubsection{The first setting: agents achieve targets}
\begin{figure}[ht]
\centering
\vspace{-3mm}
	\begin{minipage}{0.45\textwidth}
   \subfigure[\scriptsize{Average time steps}]{
    \includegraphics[width=0.45\textwidth, height=3cm]{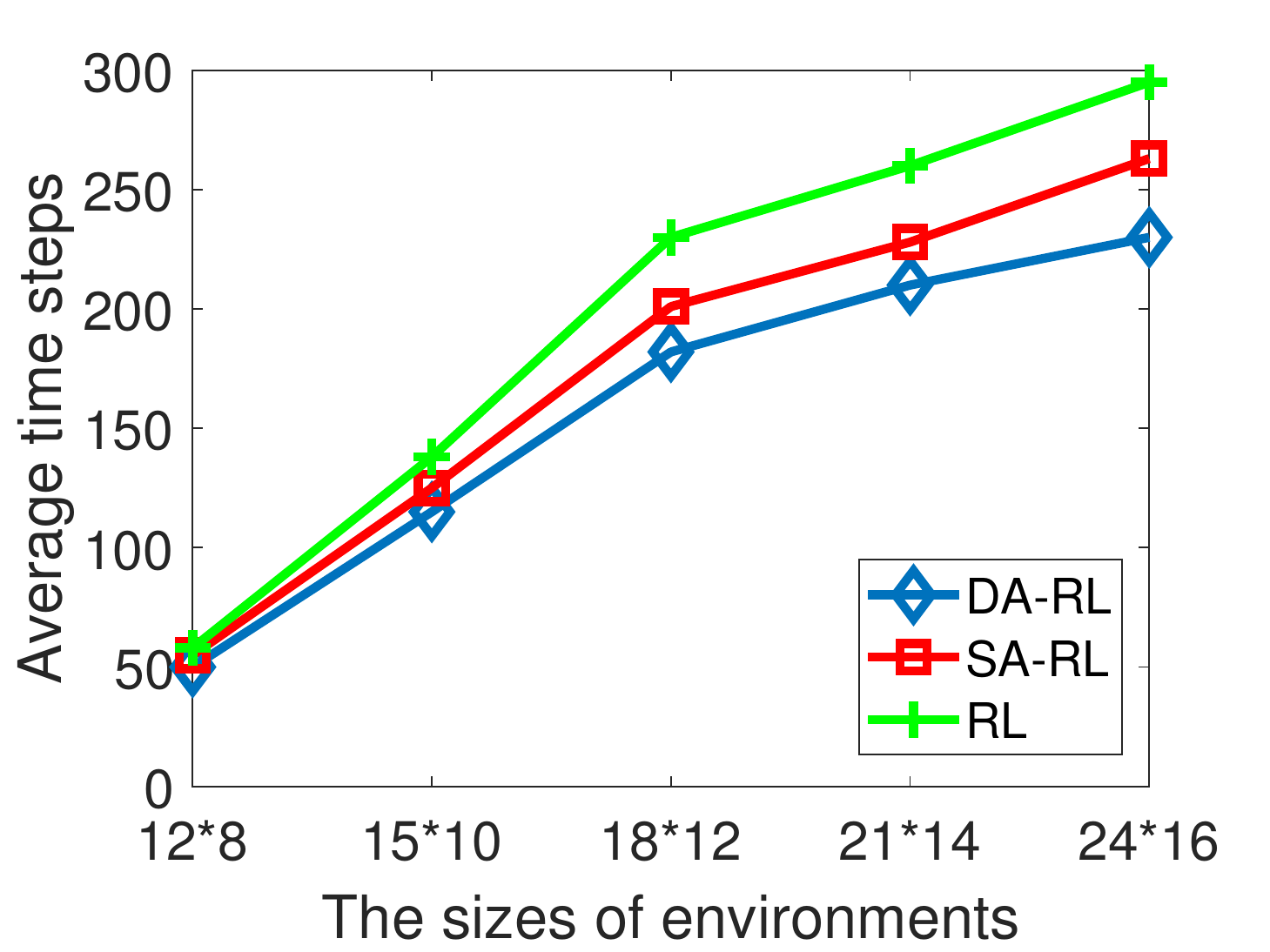}
			\label{fig:S1TimeStep}}
    \subfigure[\scriptsize{Average hits}]{
    \includegraphics[width=0.45\textwidth, height=3cm]{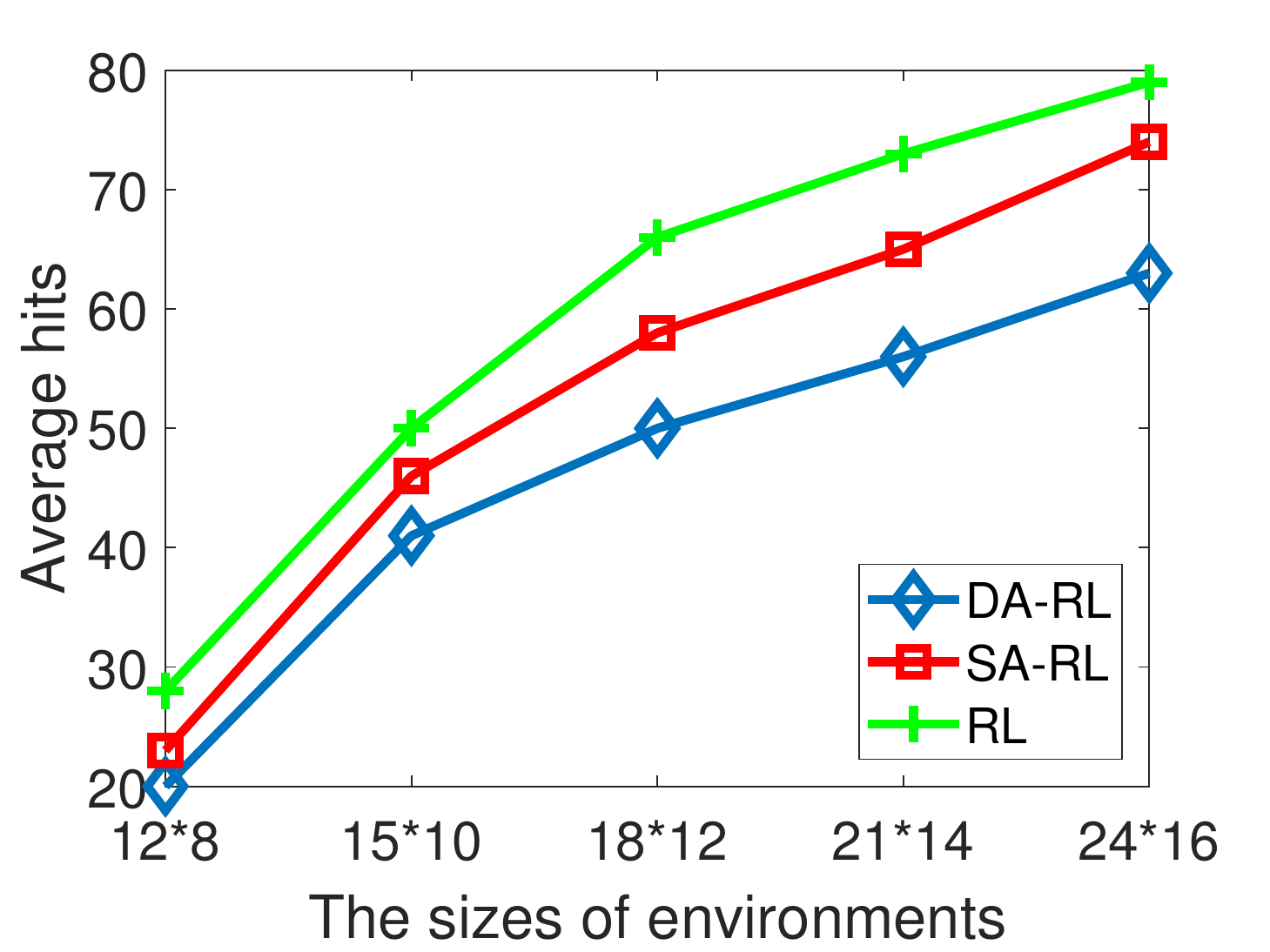}
			\label{fig:S1Hit}}\\[2ex]
    \end{minipage}
		\vspace{-3mm}
	\caption{Performance of the three methods in different sizes of environments in the first setting}
	\vspace{-1mm}
	\label{fig:S1Sizes}
\end{figure}

Fig. \ref{fig:S1Sizes} demonstrates the performance of the three methods in different sizes of environments in the first setting.
It can be seen that with the increase of the environmental size,
in all three methods, agents take more steps to achieve targets and hit more obstacles.
This is because with the increase of size,
more obstacles are there in the environments,
which inevitably increases the difficulty for agents to achieve targets.

In Fig. \ref{fig:S1TimeStep}, when the environmental size is smaller than $18\times 12$,
the performance difference among the three methods is small.
However, when the environmental size is larger than $18\times 12$,
the performance difference in average time steps enlarges up to $20\%$ between \emph{DA-RL} and \emph{RL},
and $10\%$ between \emph{DA-RL} and \emph{SA-RL}, respectively.
In addition, in Fig. \ref{fig:S1Hit}, the performance difference in average hits enlarges up to $25\%$ between \emph{DA-RL} and \emph{RL},
and $15\%$ between \emph{DA-RL} and \emph{SA-RL}, respectively.
This is because when the environmental size increases,
the number of targets and obstacles also increases.
Thus, the number of states of each agent may increase as well.
In this situation, since \emph{DA-RL} and \emph{SA-RL} enable agent advising,
their performance is better than \emph{RL}.
Moreover, as \emph{DA-RL} adopts the differential privacy technique,
agents using \emph{DA-RL} can take more advice than using \emph{SA-RL}.
Thus, the performance of \emph{DA-RL} is better than \emph{SA-RL}.

This result shows that by using agent advising,
agents can ask for advice about the positions of obstacles and targets.
Thus, agents' learning performance can be improved.
Moreover, in the \emph{DA-RL} method, an agent's advice, created in one state, can also be used by other agents in neighboring states,
which further improves those agents' learning performance.

\begin{figure}[ht]
\centering
\vspace{-5mm}
	\begin{minipage}{0.45\textwidth}
   \subfigure[\scriptsize{Average time steps}]{
    \includegraphics[width=0.45\textwidth, height=3cm]{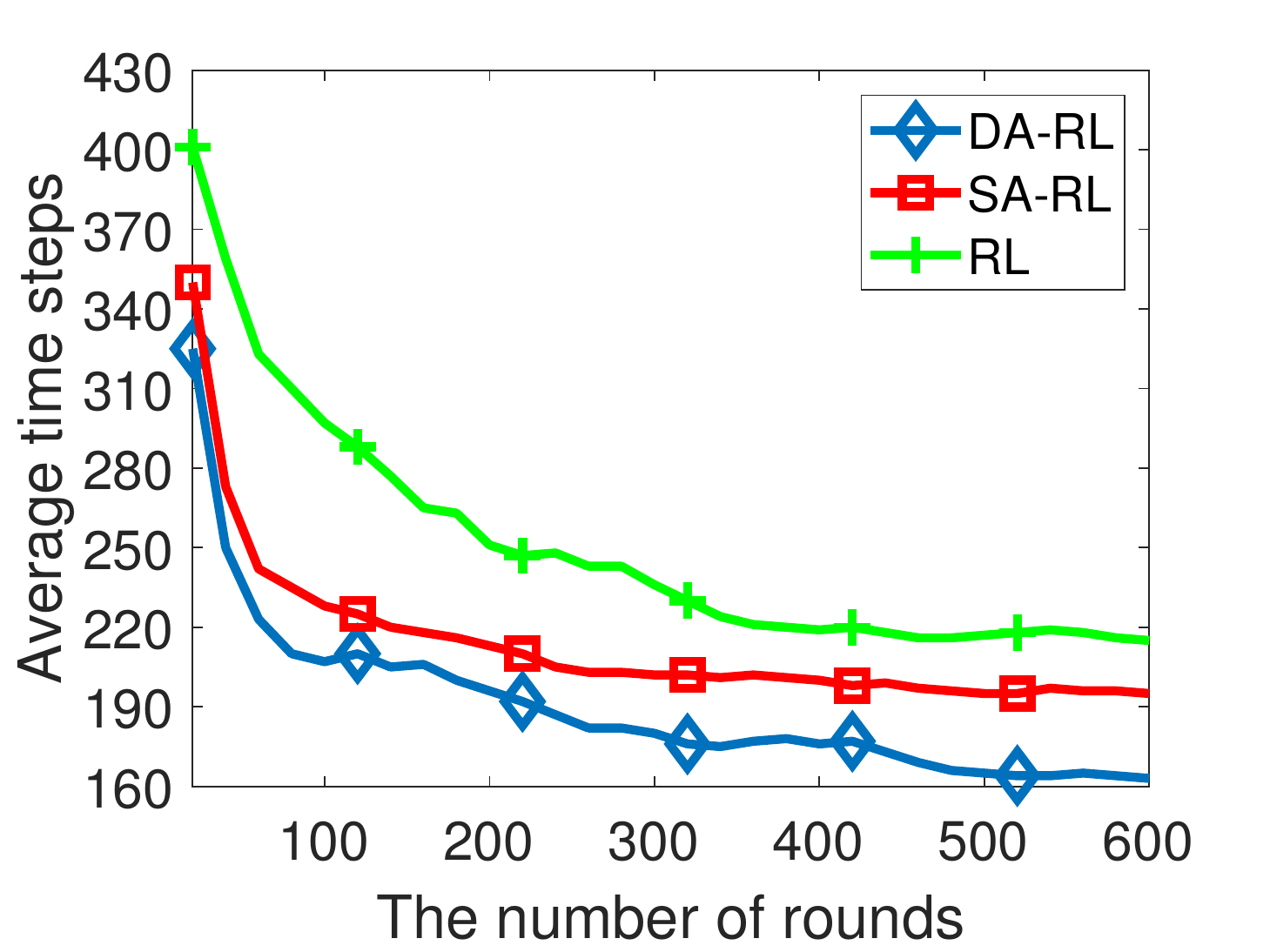}
			\label{fig:S1StepRuns}}
    \subfigure[\scriptsize{Average hits}]{
    \includegraphics[width=0.45\textwidth, height=3cm]{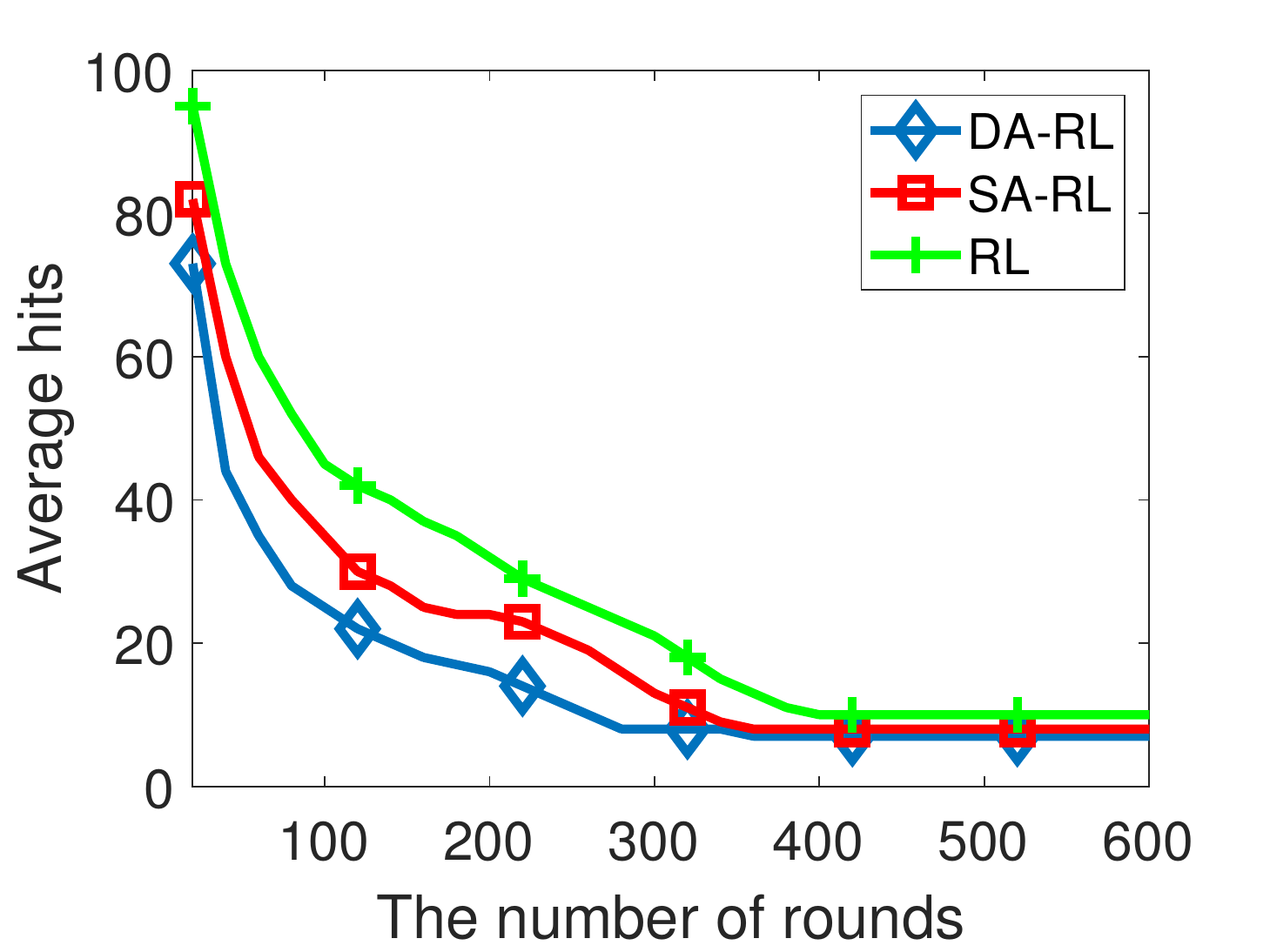}
			\label{fig:S1HitRuns}}\\[2ex]
    \end{minipage}
		\vspace{-3mm}
	\caption{Performance of the three methods as time progresses in the first setting}
	\vspace{-1mm}
	\label{fig:S1Runs}
\end{figure}

Fig. \ref{fig:S1Runs} shows the performance of the three methods as time progresses in the first setting.
The size of the environment is set to $18\times 12$;
the number of agents is set to $4$;
the number of targets is set to $40$;
and the number of obstacles is set to $25$.
In Figs. \ref{fig:S1StepRuns}, \emph{DA-RL} and \emph{SA-RL} methods continuously outperform \emph{RL} method.
Moreover, \emph{DA-RL} and \emph{SA-RL} methods converge faster than \emph{RL} method.
This can be explained by the fact that
agents using \emph{DA-RL} and \emph{SA-RL} methods can ask for advice between each other,
which can significantly improve the learning performance, especially in early stages.
In Fig. \ref{fig:S1HitRuns},
the number of average hits are almost the same among the three methods in late stages.
This is because the positions of obstacles are fixed and agents using these methods have the learning ability.
As time progresses, the positions of obstacles can be gradually memorized by agents.
Thus, the number of hits decreases and finally converges to a stable point.

%

\subsubsection{The second setting: agents achieve dynamically generated targets}
\begin{figure}[ht]
\centering
\vspace{-5mm}
	\begin{minipage}{0.45\textwidth}
   \subfigure[\scriptsize{Average time steps}]{
    \includegraphics[width=0.45\textwidth, height=3cm]{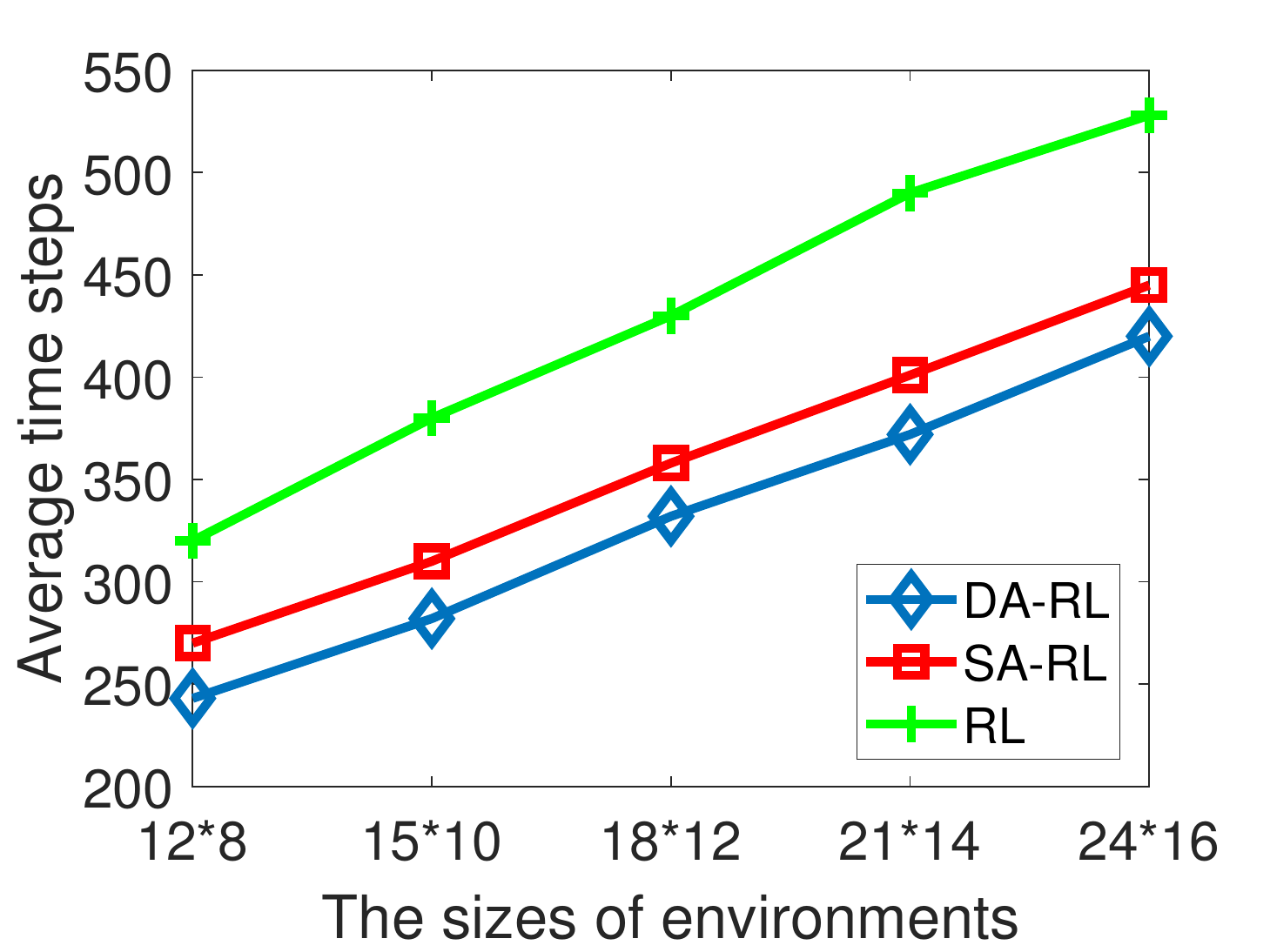}
			\label{fig:S2TimeStep}}
    \subfigure[\scriptsize{Average hits}]{
    \includegraphics[width=0.45\textwidth, height=3cm]{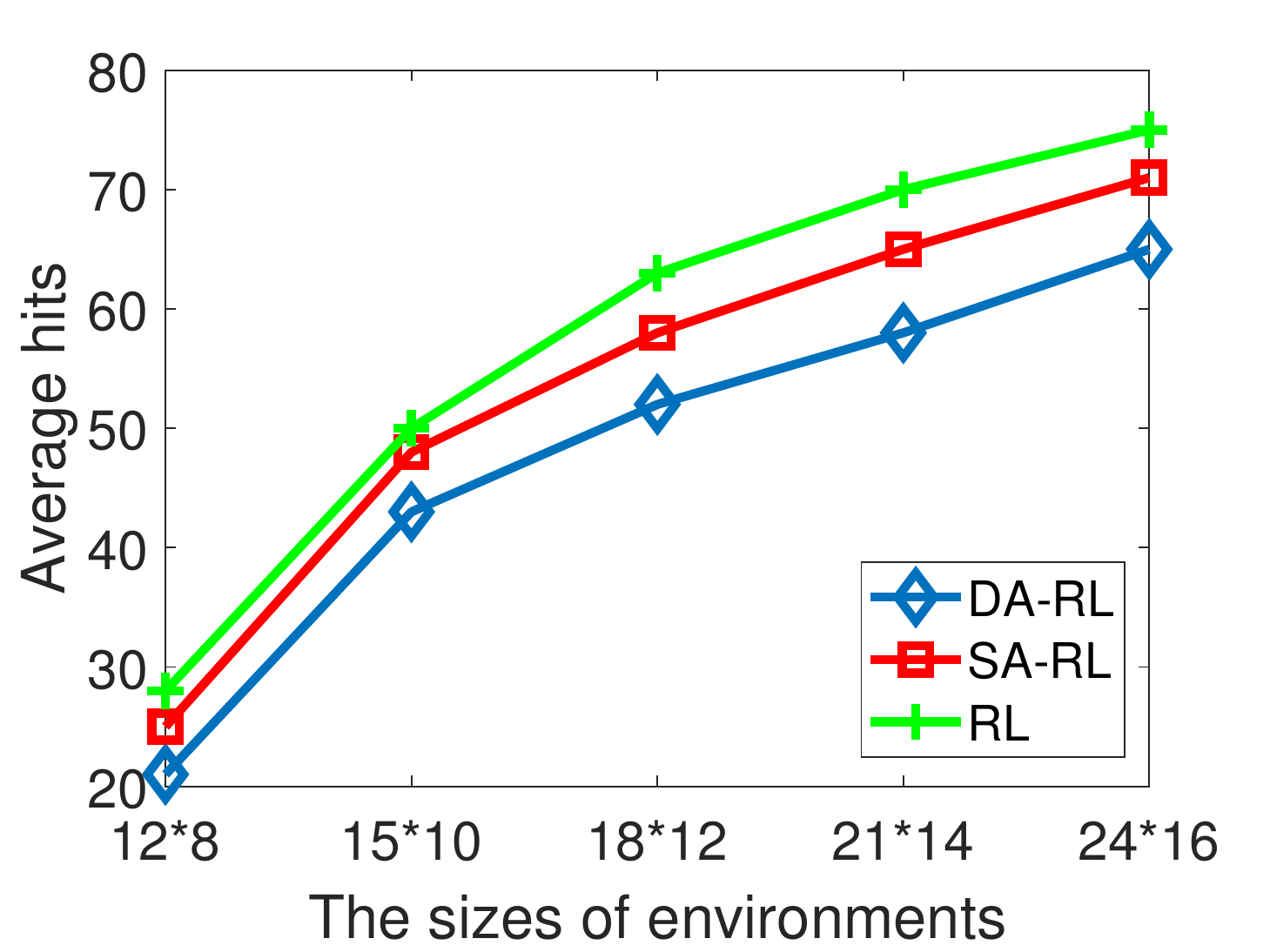}
			\label{fig:S2Hit}}\\[2ex]
    \end{minipage}
		\vspace{-3mm}
	\caption{Performance of the three methods in different sizes of environments in the second setting}
	\vspace{-1mm}
	\label{fig:S2Sizes}
\end{figure}

Fig. \ref{fig:S2Sizes} demonstrates the performance of the three methods in different sizes of environments in the second setting.
It can be seen that the proposed \emph{DA-RL} method outperforms the other two methods.

Compared to the first setting (Fig. \ref{fig:S1TimeStep}), all the three methods need more time steps to achieve targets in the second setting (Fig. \ref{fig:S2TimeStep}).
This is because in the second setting, new targets are dynamically introduced.
It is unavoidable to take more time steps to achieve these new targets.

The average number of hits in the three methods in the second setting (Fig. \ref{fig:S2Hit}) is almost the same as that in the first setting (Fig. \ref{fig:S1Hit}).
In both first and second settings, the number and positions of obstacles are fixed.
Therefore, the positions of obstacles can be learned by agents.
Agents then can avoid the obstacles to some extent.

\begin{figure}[ht]
\vspace{-5mm}
\centering
	\begin{minipage}{0.45\textwidth}
   \subfigure[\scriptsize{Average time steps}]{
    \includegraphics[width=0.45\textwidth, height=3cm]{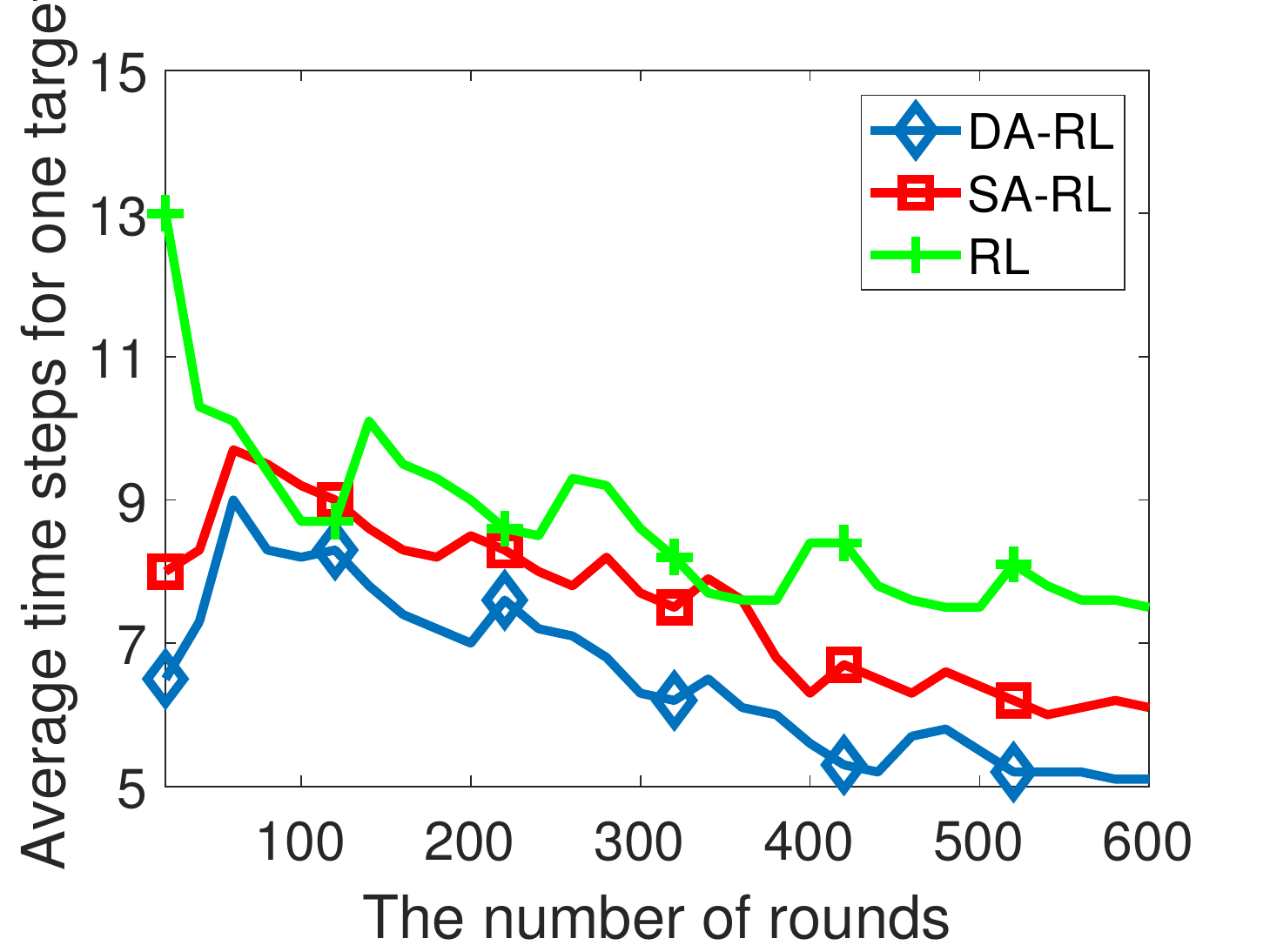}
			\label{fig:S2StepRuns}}
    \subfigure[\scriptsize{Average hits}]{
    \includegraphics[width=0.45\textwidth, height=3cm]{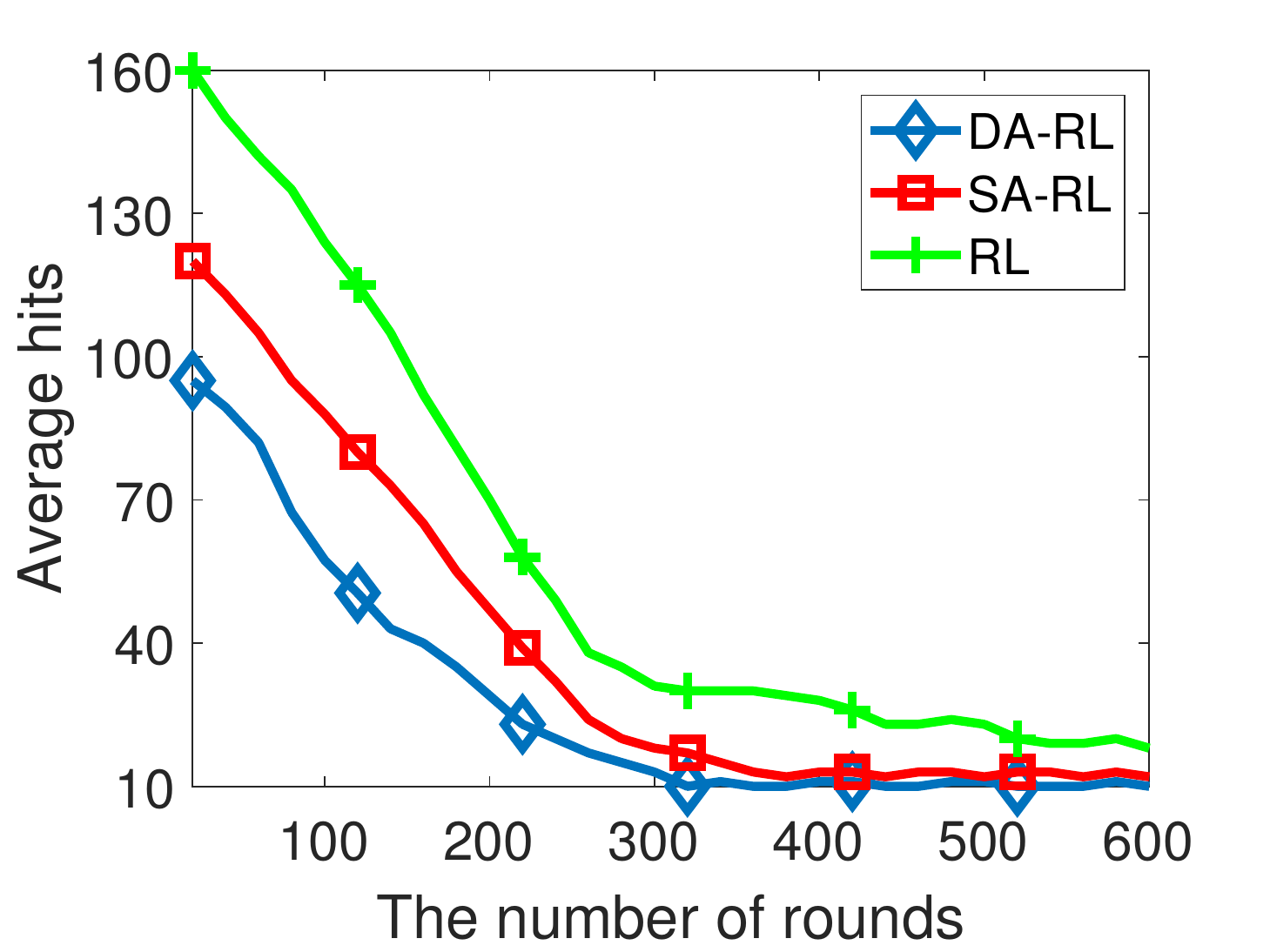}
			\label{fig:S2HitRuns}}\\[2ex]
    \end{minipage}
		\vspace{-3mm}
	\caption{Performance of the three methods as time progresses in the second setting}
	\vspace{-1mm}
	\label{fig:S2Runs}
\end{figure}

Fig. \ref{fig:S2Runs} shows the performance of the three methods as time progresses in the second setting.
The size of the environment is set to $18\times 12$;
the number of agents is set to $4$;
the number of targets is set to $40$;
and the number of obstacles is set to $25$.
In Fig. \ref{fig:S2StepRuns}, as time progresses, the average number of time steps for one target in the three methods decreases in a fluctuant manner.
This is caused by the introduction of new targets.
The introduction of new targets implies that new knowledge is introduced in environments.
Thus, agents need time to learn the new knowledge,
which results in the increase of the number of time steps.
In Fig. \ref{fig:S2HitRuns}, as time progresses, the average number of hits gradually decreases in the three methods,
because the positions of obstacles can be learned and thus obstacles can be avoided.


\subsubsection{The third setting: agents achieve dynamically moving targets}
\begin{figure}[ht]
\centering
\vspace{-5mm}
	\begin{minipage}{0.45\textwidth}
   \subfigure[\scriptsize{Average time steps}]{
    \includegraphics[width=0.45\textwidth, height=3cm]{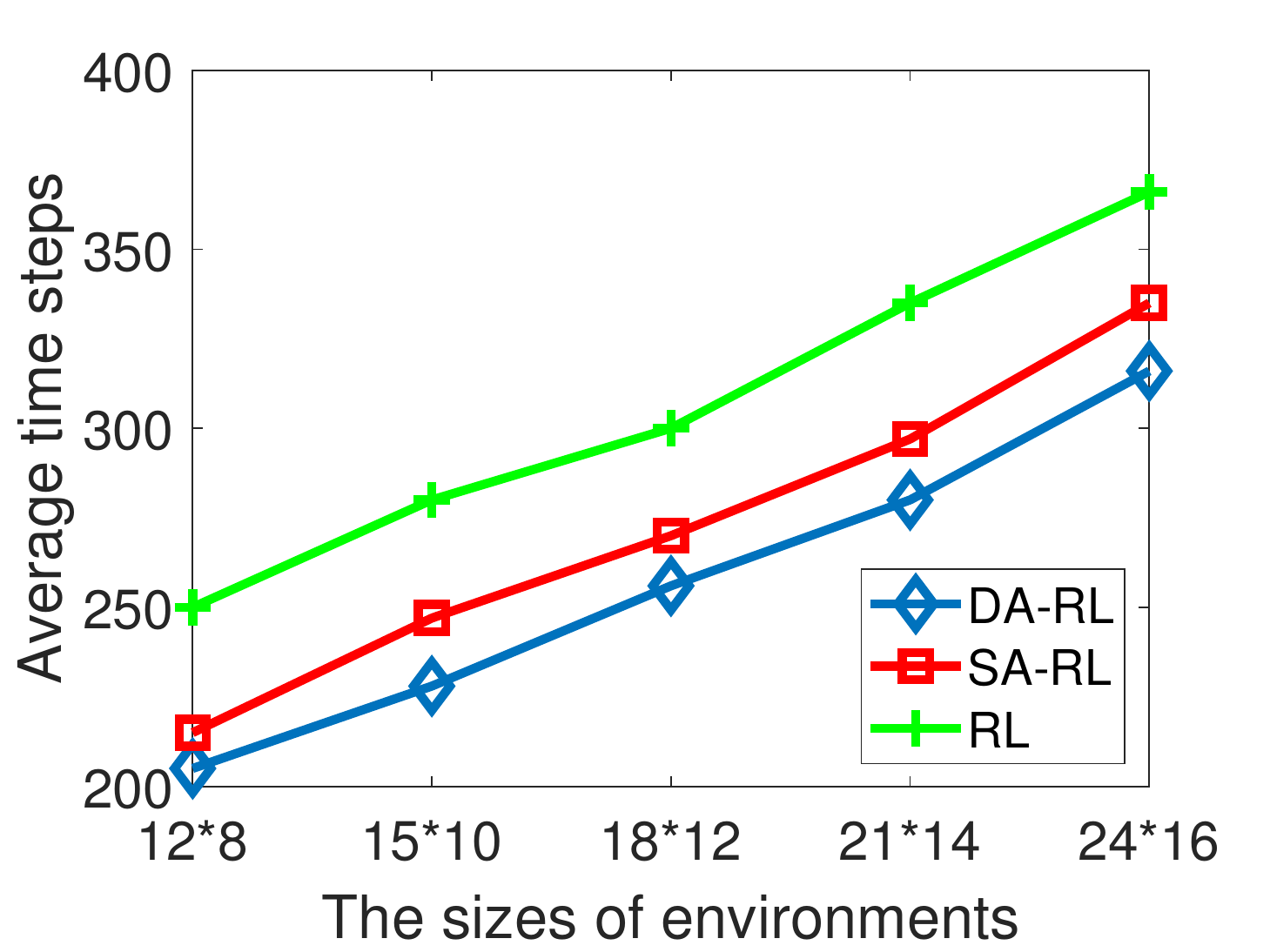}
			\label{fig:S3TimeStep}}
    \subfigure[\scriptsize{Average hits}]{
    \includegraphics[width=0.45\textwidth, height=3cm]{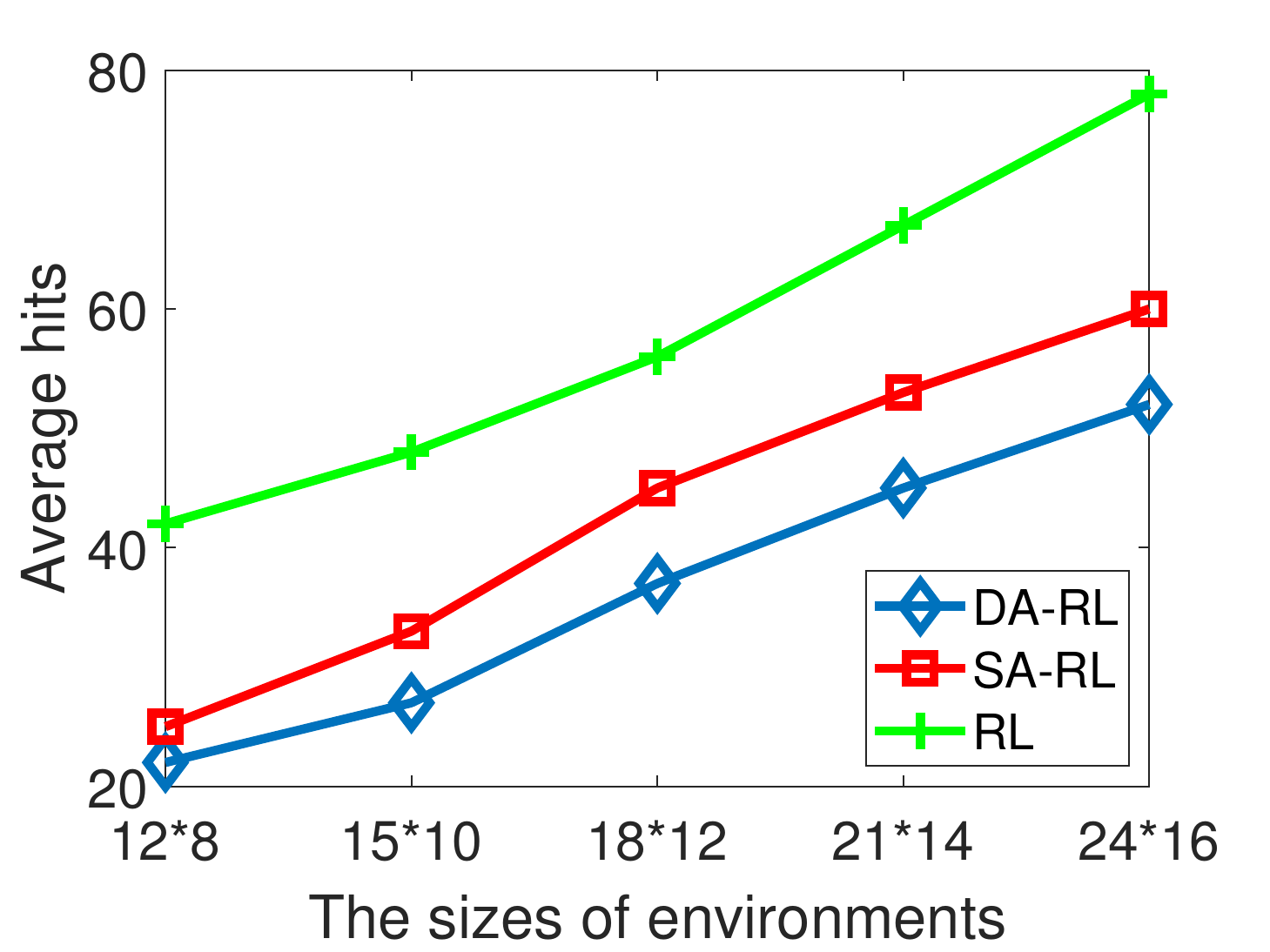}
			\label{fig:S3Hit}}\\[2ex]
    \end{minipage}
		\vspace{-3mm}
	\caption{Performance of the three methods in different sizes of environments in the third setting}
	\vspace{-1mm}
	\label{fig:S3Sizes}
\end{figure}

Fig. \ref{fig:S3Sizes} demonstrates the performance of the three methods in different sizes of environments in the third scenario.
Again, the proposed \emph{DA-RL} method outperforms the other two methods.
Comparing Fig. \ref{fig:S3Sizes} and Fig. \ref{fig:S2Sizes},
the performance trend of the three methods in the third setting is similar to the second setting.
This is because both the second and third settings are dynamic,
which means that agents need more time steps to achieve targets than in the first setting.

\begin{figure}[ht]
\centering
\vspace{-5mm}
	\begin{minipage}{0.45\textwidth}
   \subfigure[\scriptsize{Average time steps}]{
    \includegraphics[width=0.45\textwidth, height=3cm]{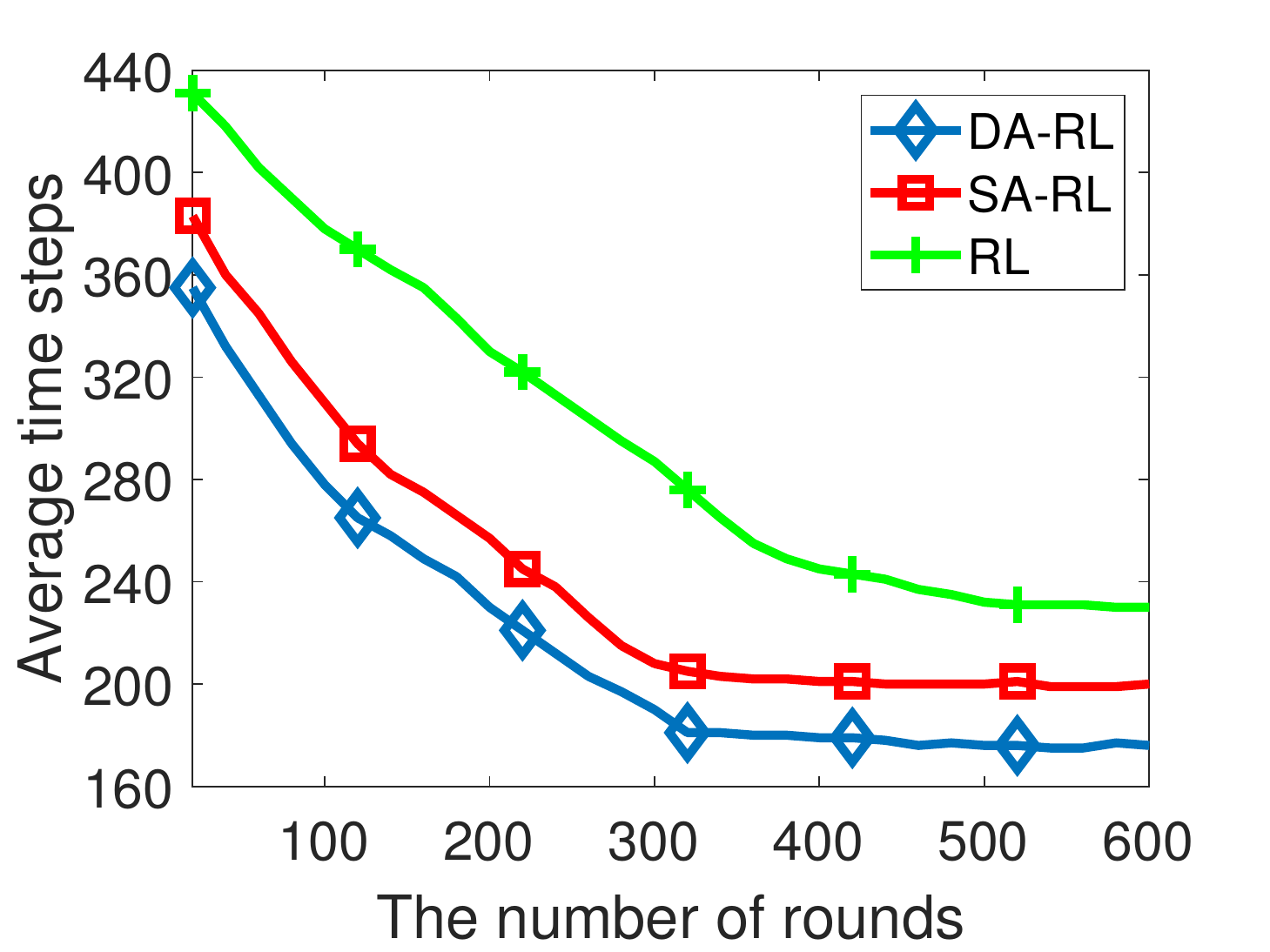}
			\label{fig:S3StepRuns}}
    \subfigure[\scriptsize{Average hits}]{
    \includegraphics[width=0.45\textwidth, height=3cm]{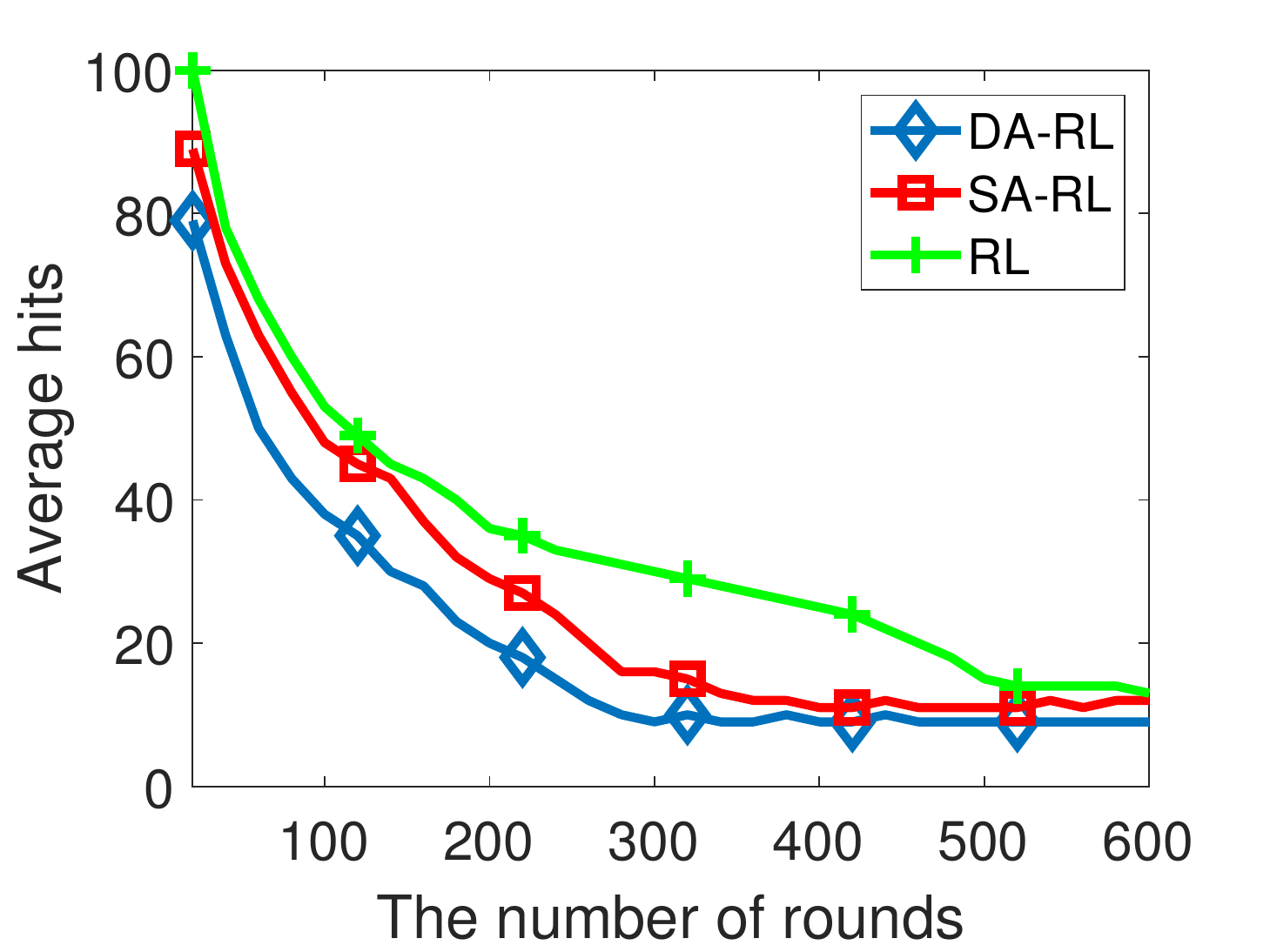}
			\label{fig:S3HitRuns}}\\[2ex]
    \end{minipage}
		\vspace{-3mm}
	\caption{Performance of the three methods as time progresses in the third setting}
	\vspace{-1mm}
	\label{fig:S3Runs}
\end{figure}

Fig. \ref{fig:S3Runs} shows the performance of the three methods as time progresses in the third setting.
The size of the environment is set to $18\times 12$;
the number of agents is set to $4$;
the number of targets is set to $40$;
and the number of obstacles is set to $25$.
Unlike the second scenario (Fig. \ref{fig:S2StepRuns}), there is no fluctuation in the three methods in the third setting (Fig. \ref{fig:S3StepRuns}).
This is because in the third setting, although targets are moving, the number of targets does not increase.
For example, there are $3$ targets in an environment.
In the third setting, the experiment finishes when all the $3$ targets are achieved.
In the second setting, however, there may be new targets introduced.
The experiment will not finish until all the $3$ and the new targets are achieved.
If new targets are generated constantly,
the experiment is hard to finish.
This means that in the second setting,
the completion of an experiment heavily depends on the generation probability of new targets.
Hence, there is more fluctuation in the second setting than in the third setting.

\subsection{Experimental results of scenario 2}
\begin{figure}[ht]
\centering
\vspace{-5mm}
	\begin{minipage}{0.45\textwidth}
   \subfigure[\scriptsize{Average reward with different number of agents}]{
    \includegraphics[width=0.45\textwidth, height=3cm]{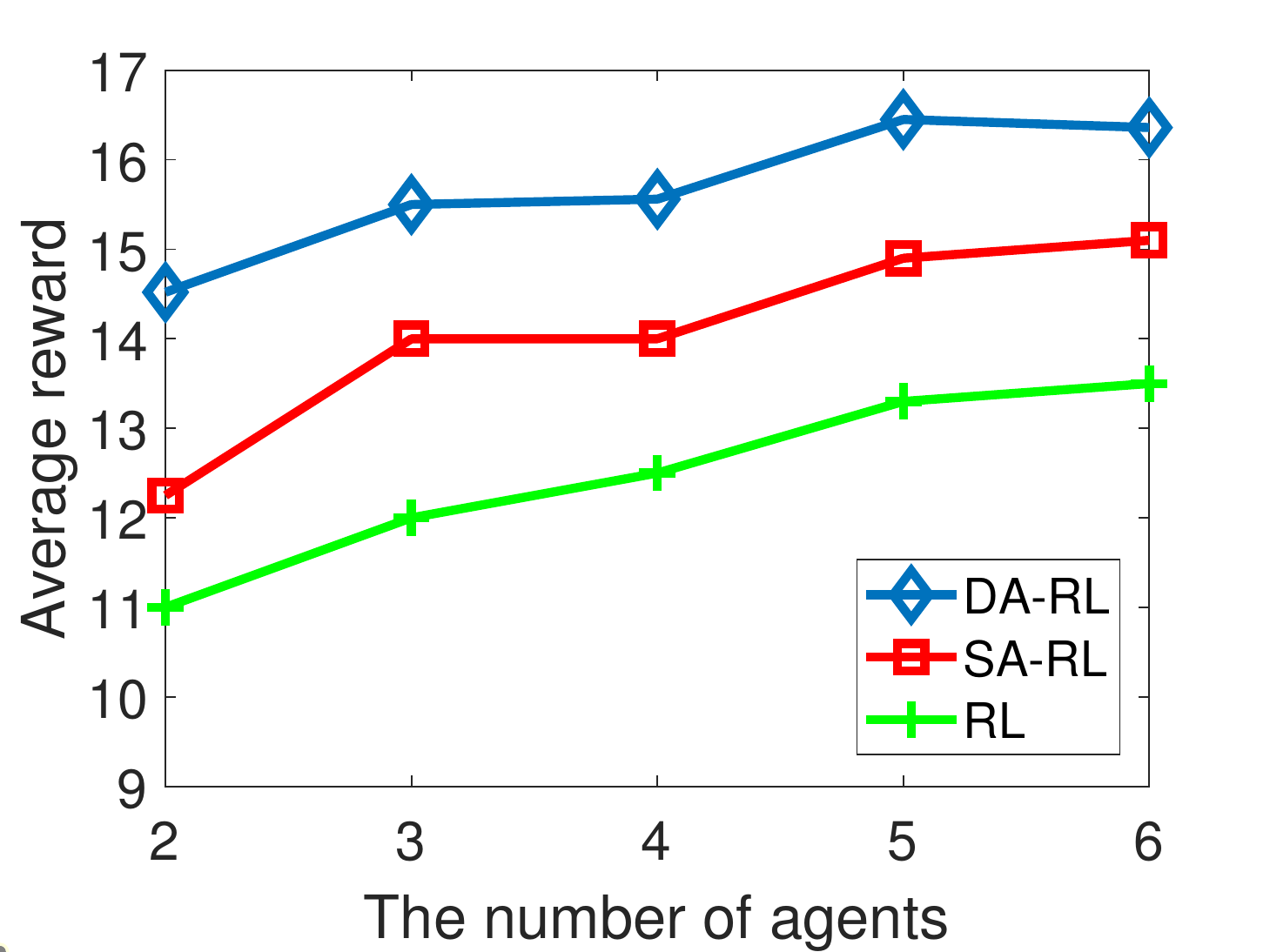}
			\label{fig:S4Scale}}
    \subfigure[\scriptsize{Average reward with different number of item types}]{
    \includegraphics[width=0.45\textwidth, height=3cm]{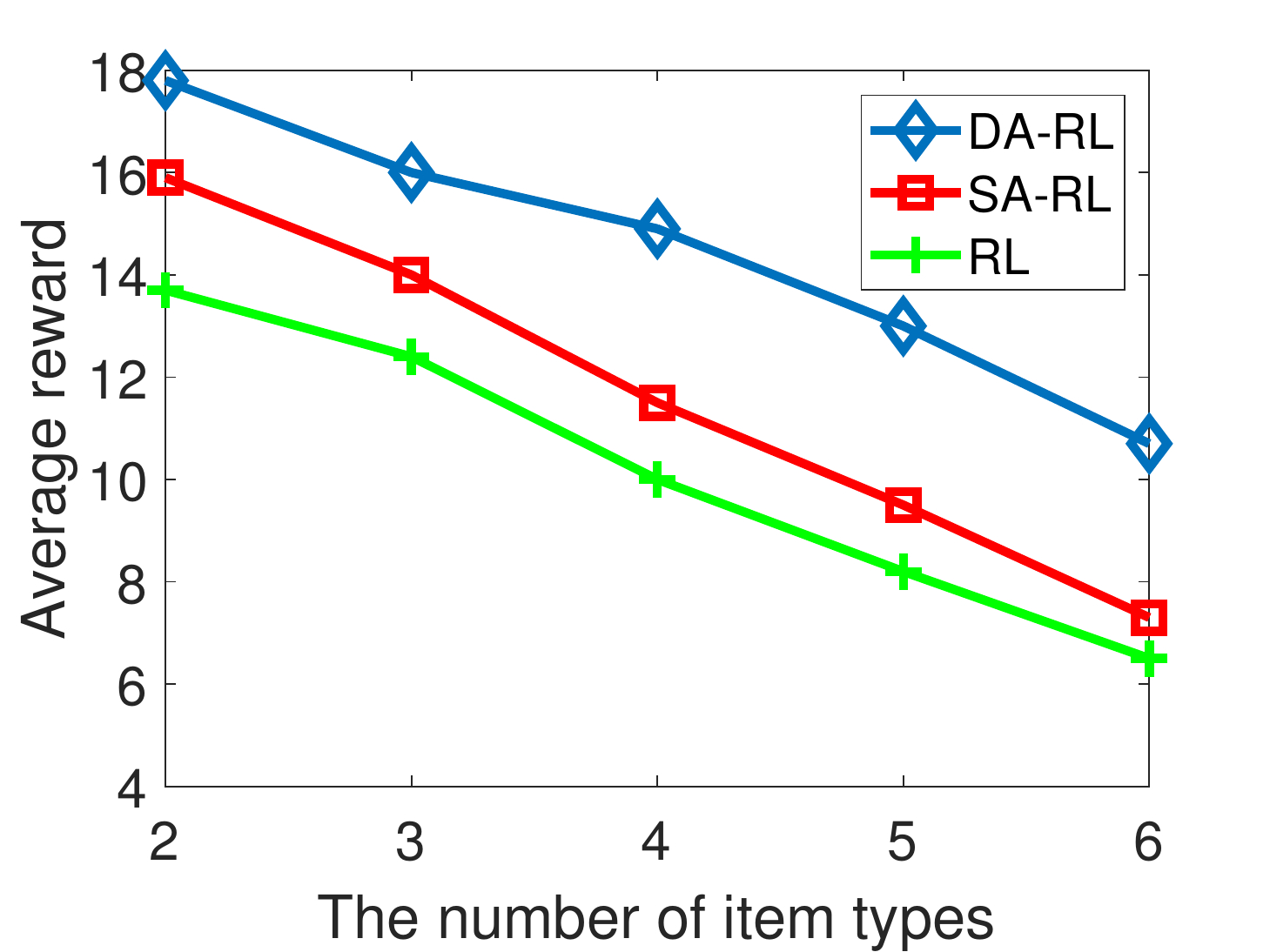}
			\label{fig:S4Item}}\\[2ex]
   \subfigure[\scriptsize{Average reward with different size of stocks}]{
    \includegraphics[width=0.45\textwidth, height=3cm]{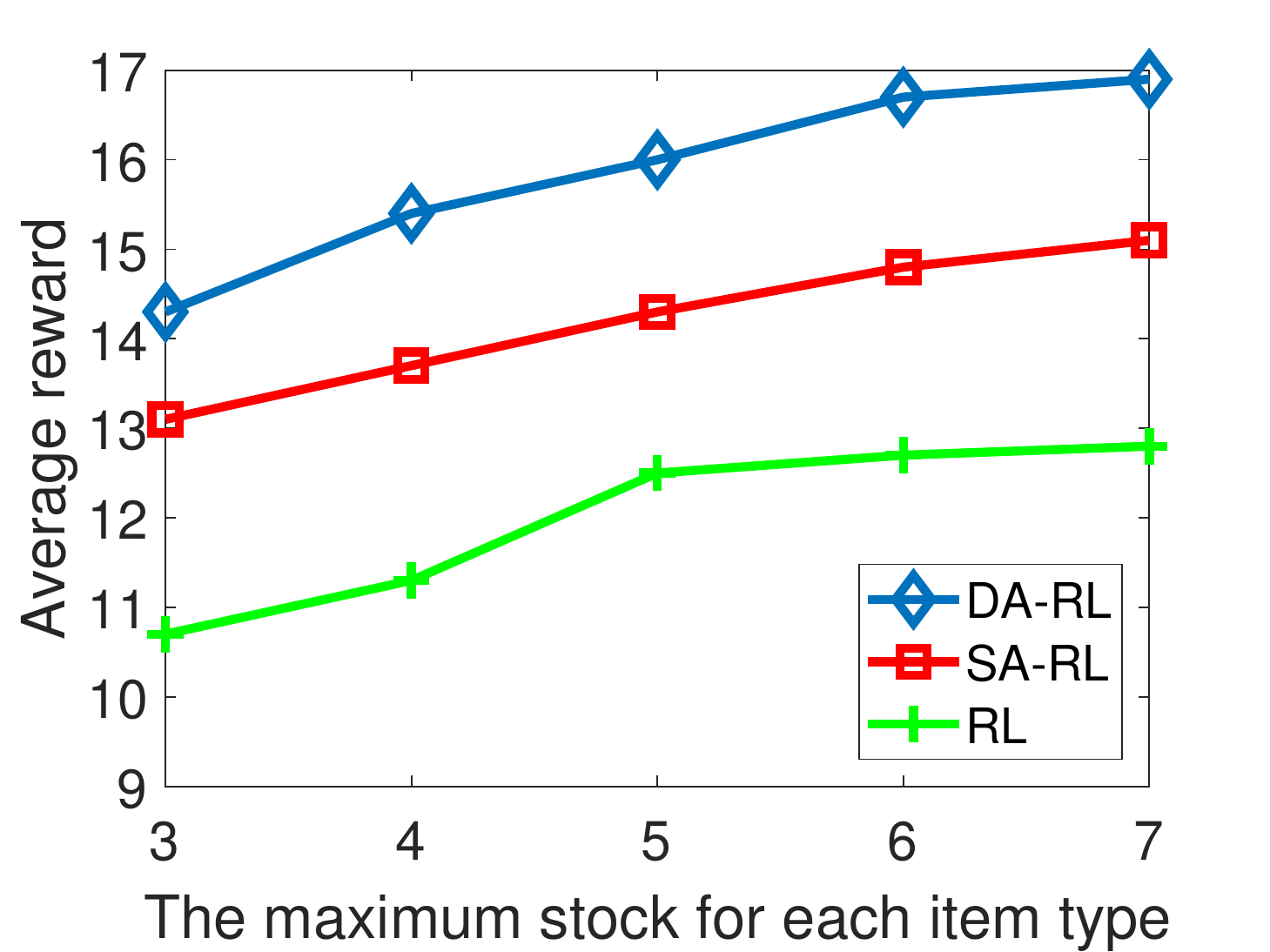}
			\label{fig:S4Stock}}
    \subfigure[\scriptsize{Average reward with different processing probabilities}]{
    \includegraphics[width=0.45\textwidth, height=3cm]{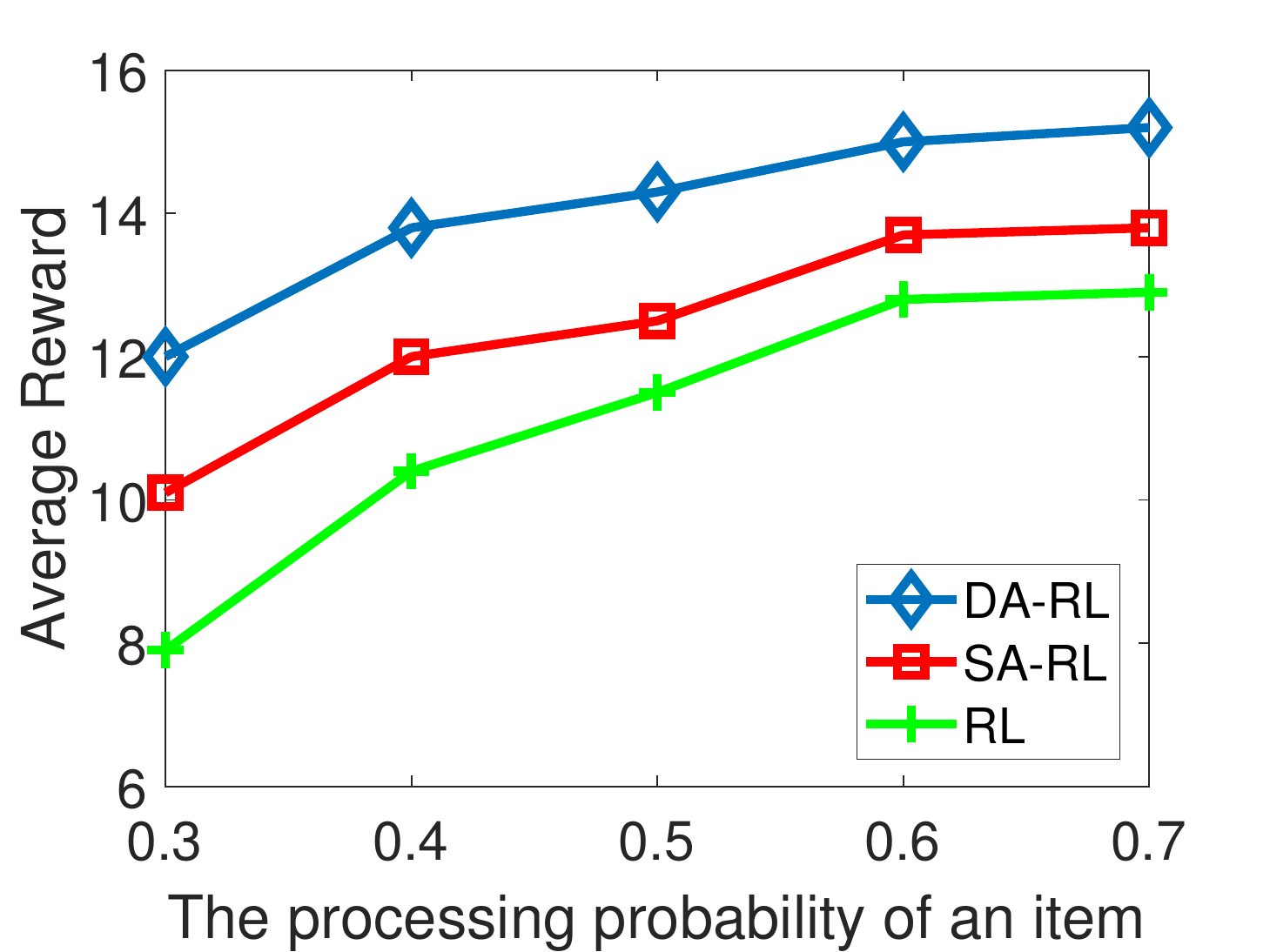}
			\label{fig:S4Process}}\\[2ex]
    \end{minipage}
		\vspace{-3mm}
	\caption{Performance of the three methods}
	\vspace{-1mm}
	\label{fig:S4}
\end{figure}
Fig. \ref{fig:S4Scale} demonstrates the performance of the three methods with different number of agents, 
where the number of item types is set to $3$, the maximum stock of each item type is set to $5$, 
and the probability of processing an item is set to $0.5$. 
In Fig. \ref{fig:S4Scale}, with the increase of the number of agents, 
the average reward of agents in the three methods also rises.
When the number of agents increases, 
each agent can ask for advice from more agents. 
Thus, agents can learn faster and receive more reward. 
However, it should also be noted that 
when the number of agents is larger than $5$, 
the average reward keeps almost steady. 
This can be explained by the fact that in a given environment, 
the amount of knowledge is limited. 
Agents can share knowledge but cannot create new knowledge. 
When the number of agents is large enough to guarantee a good learning speed, 
increasing the number of agents does not improve agents' learning speed or increase agents' reward.

Fig. \ref{fig:S4Item} shows the performance of the three methods with different number of item types, 
where the number of agents is set to $4$, the maximum stock of each item type is set to $5$, 
and the probability of processing an item is set to $0.5$. 
In Fig. \ref{fig:S4Item}, with the increase of the number of item types, 
the average reward of agents in the three methods reduces gradually. 
When the number of item types increases, 
the dimension of states of each agent also increases. 
This decreases the learning speed of agents, 
which reduces the average reward of agents 
as agents have more opportunity to make non-optimal decisions. 

Fig. \ref{fig:S4Stock} displays the performance of the three methods with different maximum stock, 
where the number of agents is set to $4$, the number of item types is set to $3$, 
and the probability of processing an item is set to $0.5$. 
In Fig. \ref{fig:S4Stock}, with the increase of the maximum stock, 
the average reward of agents in the three methods increases gracefully. 
When the stock capacity of agents increases, 
agents prefer to store new items instead of redistributing them. 
This preference will augment agents' average reward, 
because, based on the experimental setting, storing an item can incur more reward than redistributing it.

Fig. \ref{fig:S4Process} exhibits the performance of the three methods with different probabilities of processing an item, 
where the number of agents is set to $4$, the number of item types is set to $4$, 
and the maximum stock of each item type is set to $5$. 
In Fig. \ref{fig:S4Process}, with the increase of processing probability, 
the average reward of agents in the three methods rises. 
When the processing probability increases, 
items are processed faster and hence, the stock amount of each agent reduces. 
This improves the average reward of agents, 
as the reward is based on stock amount in the experimental setting. 
However, it should also be noted that 
when the processing probability is too large, e.g., larger than $0.6$, 
the difference between the three methods reduces. 
As mentioned before, a large processing probability implies a small stock amount. 
Since the state of an agent is based on the stock amount, i.e., the number of items of each item type, 
small stock amount means a small number of states. 
Because each agent observes only a small number of states, 
after a short learning process, each agent is confident in these states 
and does not ask for advice any more. 
If agents do not ask for advice, 
the \emph{DA-RL} and \emph{SA-RL} methods are the same as \emph{RL} method.
Hence, they have similar results.

In the above four situations, the proposed \emph{DA-RL} method outperforms the \emph{SA-RL} and \emph{RL} methods, 
though the performance of the three methods has a similar trend. 
According to the experimental results, the advantage of differential advising has been empirically proven.

\begin{figure}[ht]
\centering
\vspace{-5mm}
	\begin{minipage}{0.45\textwidth}
   \subfigure[\scriptsize{Average reward as learning progresses in a simple environment}]{
    \includegraphics[width=0.45\textwidth, height=3cm]{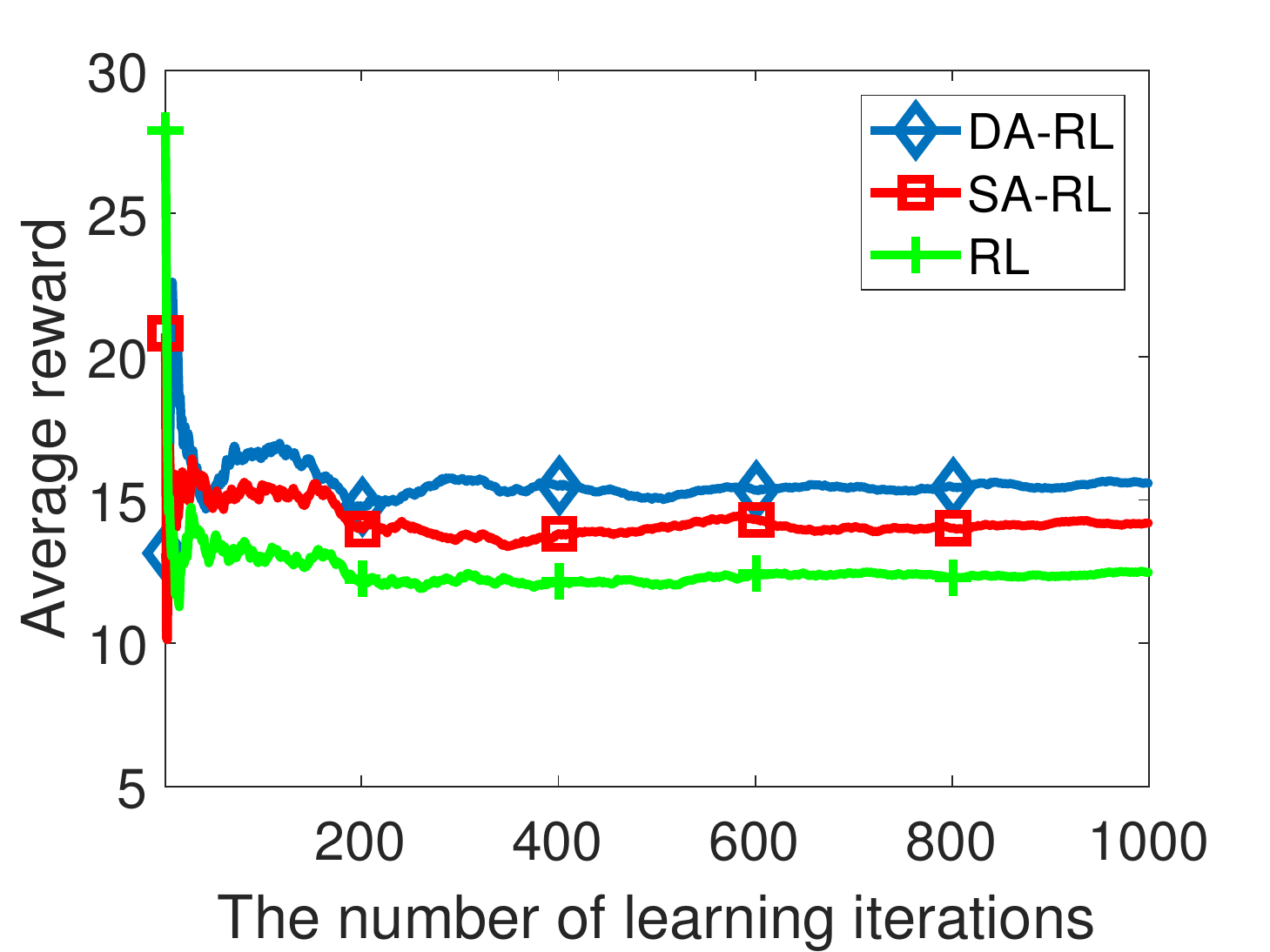}
			\label{fig:S4Runs1}}
    \subfigure[\scriptsize{Average reward as learning progresses in a complex environment}]{
    \includegraphics[width=0.45\textwidth, height=3cm]{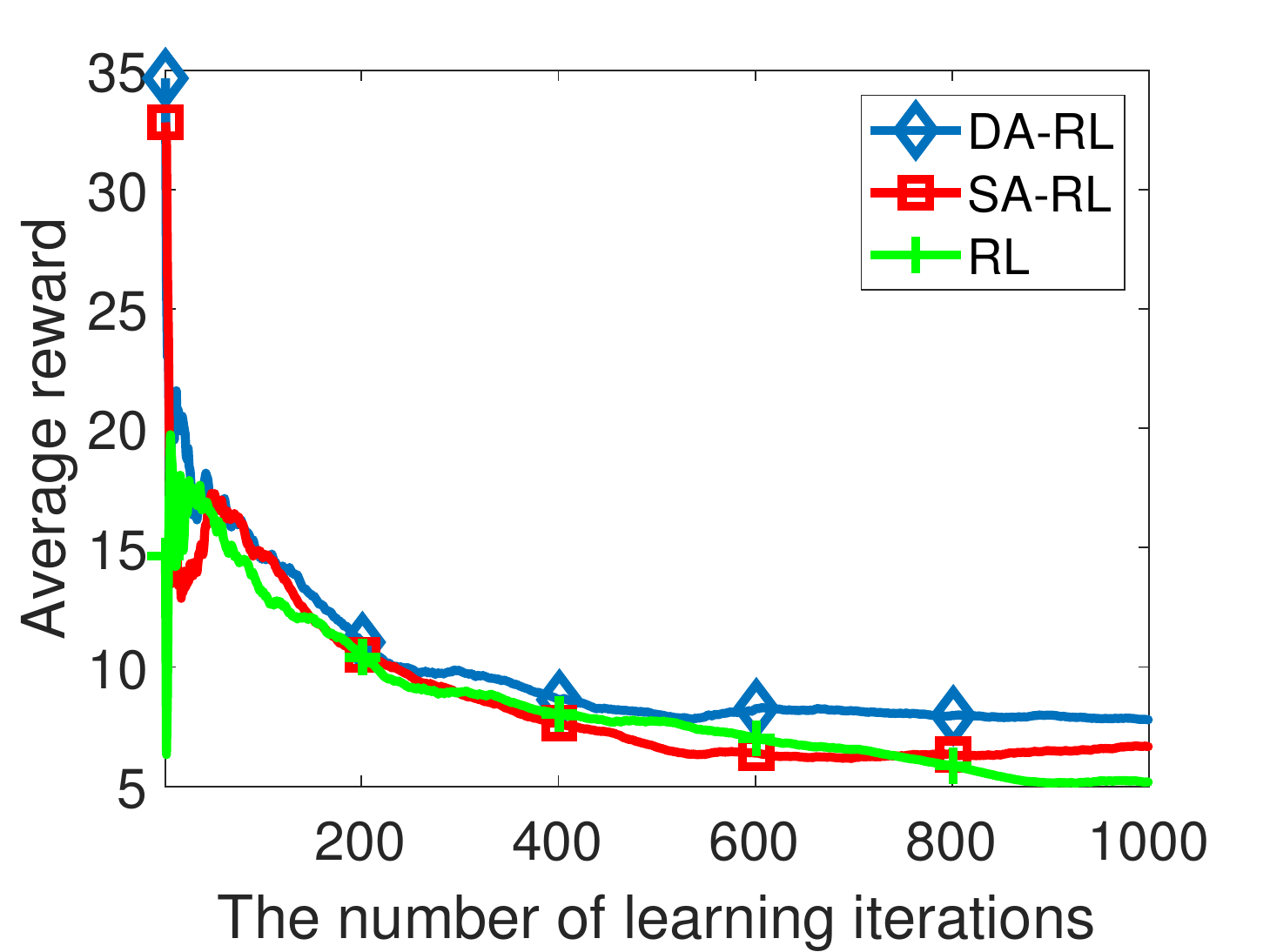}
			\label{fig:S4Runs2}}\\[2ex]
    \end{minipage}
		\vspace{-3mm}
	\caption{Performance of the three methods as learning progresses}
	\vspace{-1mm}
	\label{fig:S4Runs}
\end{figure}
Fig. \ref{fig:S4Runs} demonstrates the performance of the three methods as learning progresses. 
Fig. \ref{fig:S4Runs1} shows the results in a simple environment. 
The number of agents is set to $3$, the number of item types is set to $2$, 
the maximum stock of each item type is set to $3$, 
and the probability of processing an item is set to $0.5$. 
Fig. \ref{fig:S4Runs2} shows the results in a complex environment. 
The number of agents is set to $3$, the number of item types is set to $6$,
the maximum stock of each item type is set to $7$,
and the probability of processing an item is set to $0.5$. 

In Fig. \ref{fig:S4Runs1}, the three methods converge in very early stages (around $200$ learning iterations). 
In comparison, in Fig. \ref{fig:S4Runs2}, the three methods converge much slower, 
where \emph{DA-RL} converges at around $500$ learning iterations, \emph{SA-RL} converges at around $600$ iterations, 
and \emph{RL} converges at around $900$ iterations. 
This can be explained by the fact that in the complex environment, 
the number of states is much more than the number of states in the simple environment.
As more states typically imply more knowledge, 
agents have to take more learning iterations to learn the knowledge. 
Therefore, the three methods converge slower in the complex environment than in the simple environment. 
However, the proposed \emph{DA-RL} method in the complex environment converges faster than the \emph{SA-RL} and \emph{RL} methods. 
This result demonstrates that by using differential advising, 
agents in the proposed \emph{DA-RL} method learn faster than the other two methods.

\vspace{-1mm}
\subsection{Discussion and summary}
\subsubsection{The metric used in the second setting of the first scenario}
In the second setting of the first scenario (Fig. \ref{fig:S2StepRuns}), we use the metric, \emph{average number of time steps for one target},
instead of the metric, \emph{average number of time steps for all targets},
used in the first and third scenarios.

The average number of time steps for all targets in the second scenario does not converge as time progresses in the three methods.
This is because agents in the three methods have learning ability.
The aim of learning is to maximize agents' long-term accumulated rewards.
In the experiments, an agent can receive a positive reward only when the agent achieves a target.
In the second setting, new targets are introduced dynamically.
This motivates agents to stay in the experiment to keep achieving targets so as to increase their accumulated rewards.
Therefore, the average number of time steps for achieving all the targets in the second scenario increases gradually.
However, this does not mean that the three methods do not converge in other aspects in the second setting.
Although the number of time steps in the three methods increases,
the number of achieved targets also increases.
Since the average number of time steps for one target is
the ratio between the total number of time steps and the total number of achieved targets,
the increased number of achieved targets can offset the increased number of time steps.
Thus, the average number of time steps for one target still converges (Fig. \ref{fig:S2StepRuns}).



\subsubsection{Performance of the three methods}
According to the experimental results in the first scenario, the average time steps used in the \emph{DA-RL} method is fewer than the other two methods in all the three settings:
about $15\%\sim 20\%$ and $40\%$ fewer than the \emph{SA-RL} and \emph{RL} methods, respectively.
Moreover, the \emph{DA-RL} converges about $10\%$ and $20\%$ faster than the \emph{SA-RL} and \emph{RL} methods, respectively.
In the second scenario, the average reward of agents in the \emph{DA-RL} method is about $10\%\sim 15\%$ and $20\%\sim 30\%$ 
more than the \emph{SA-RL} and \emph{RL} methods, respectively.
In addition, the \emph{DA-RL} method converges about $15\%$ and $40\%$ faster than the \emph{SA-RL} and \emph{RL} methods, respectively.
In summary, the proposed \emph{DA-RL} method outperforms the other two methods in both scenarios.

\vspace{-2mm}
\section{Conclusion and Future Work}\label{sec:conclusion}
This paper proposed a differential advising method for multi-agent reinforcement learning.
This method is the first to allow agents in one state to use advice created in another different state.
This method is also the first to take advantage of the differential privacy technique for agent advising to improving learning performance instead of preserving privacy.
Experimental results demonstrate that our method outperforms other benchmark methods in various aspects.

In this paper, we have an assumption that in two similar states, 
the applicability and adequacy of actions to take are similar. 
This assumption, however, are not applicable in some situations. 
For example, the good actions to take in a rainy day can become the bad actions to take in a sunny day. 
In the future, we will relax this assumption by assigning weights to dimensions of states. 
Moreover, 
our method mainly focuses on discrete states. 
It is interesting to extend our method to continuous environments.
An intuitive way for the extension is to discretize continuous states to discrete states. 
Finally, extending our method to multi-agent deep reinforcement learning is also interesting future work. 
A potential way is to allow agents to transfer experience samples as advice. 
However, if the number of transferred samples is large, 
how to select high-quality samples is a challenge to the adviser agent.



\bibliographystyle{IEEEtran}
{\small \bibliography{references}}

\begin{thebibliography}{10}
\providecommand{\url}[1]{#1}
\csname url@samestyle\endcsname
\providecommand{\newblock}{\relax}
\providecommand{\bibinfo}[2]{#2}
\providecommand{\BIBentrySTDinterwordspacing}{\spaceskip=0pt\relax}
\providecommand{\BIBentryALTinterwordstretchfactor}{4}
\providecommand{\BIBentryALTinterwordspacing}{\spaceskip=\fontdimen2\font plus
\BIBentryALTinterwordstretchfactor\fontdimen3\font minus
  \fontdimen4\font\relax}
\providecommand{\BIBforeignlanguage}[2]{{%
\expandafter\ifx\csname l@#1\endcsname\relax
\typeout{** WARNING: IEEEtran.bst: No hyphenation pattern has been}%
\typeout{** loaded for the language `#1'. Using the pattern for}%
\typeout{** the default language instead.}%
\else
\language=\csname l@#1\endcsname
\fi
#2}}
\providecommand{\BIBdecl}{\relax}
\BIBdecl

\bibitem{Silva19c}
F.~L.~D. Silva, R.~Glatt, and A.~H.~R. Costa, ``{MOO-MDP: An Object-Oriented
  Representation for Cooperative Multiagent Reinforcement Learning},''
  \emph{IEEE Trans. on Cybernetics}, vol.~49, pp. 567--579, 2019.

\bibitem{Silva17}
F.~L.~D. Silva, R.~Glatt, and A.~H.~R. Costa, ``{Simultaneously Learning and
  Advising in Multiagent Reinforcement Learning},'' in \emph{Proc. of AAMAS},
  Brazil, May 2017, pp. 1100--1108.

\bibitem{Silva18}
F.~L.~D. Silva, M.~E. Taylor, and A.~H.~R. Costa, ``{Autonomously Reusing
  Knowledge in Multiagent Reinforcement Learning},'' in \emph{Proc. of IJCAI
  2018}, Sweden, July 2018, pp. 5487--5493.

\bibitem{Sliva19b}
F.~L.~D. Silva, ``{Integrating Agent Advice and Previous Task Solutions in
  Multiagent Reinforcement Learning},'' in \emph{AAMAS}, 2019, pp. 2447--2448.

\bibitem{Silva19}
F.~L.~D. Silva and A.~H.~R. Costa, ``{A Survey on Transfer Learning for
  Multiagent Reinforcement Learning Systems},'' \emph{Journal of Artificial
  Intelligence Research}, vol.~64, pp. 645--703, 2019.

\bibitem{Ye19}
D.~Ye, T.~Zhu, W.~Zhou, and P.~S. Yu, ``{Differentially Private Malicious Agent
  Avoidance in Multiagent Advising Learning},'' \emph{IEEE Transactions on
  Cybernetics}, vol.~50, no.~10, pp. 4214--4227, 2020.

\bibitem{Dwork06}
C.~Dwork, ``Differential privacy,'' in \emph{Proc. of ICALP}, 2006, pp. 1--12.

\bibitem{Zhu17}
T.~Zhu, G.~Li, W.~Zhou, and P.~S. Yu, ``{Differentially private data publishing
  and analysis: A survey},'' \emph{IEEE Transactions on Knowledge and Data
  Engineering}, vol.~29, no.~8, pp. 1619--1638, 2017.

\bibitem{Torrey13}
L.~Torrey and M.~E. Taylor, ``{Teaching on A Budget: Agents Advising Agents in
  Reinforcement Learning},'' in \emph{AAMAS}, 2013, pp. 1053--1060.

\bibitem{Amir16}
O.~Amir, E.~Kamar, A.~Kolobov, and B.~Grosz, ``{Interactive Teaching Strategies
  for Agent Training},'' in \emph{Proc. of IJCAI}, 2016, pp. 804--811.

\bibitem{Wang18b}
Y.~Wang, W.~Lu, J.~Hao, J.~Wei, and H.~Leung, ``{Efficient Convention Emergence
  through Decoupled Reinforcement Social Learning with Teacher-Student
  Mechanism},'' in \emph{Proc. of AAMAS}, 2018, pp. 795--803.

\bibitem{Omid19}
S.~Omidshafiei, D.~Kim, M.~Liu, G.~Tesauro, M.~Riemer, C.~Amato, M.~Campbell,
  and J.~P. How, ``{Learning to Teach in Cooperative Multiagent Reinforcement
  Learning},'' in \emph{AAAI}, 2019, pp. 6128--6136.

\bibitem{Zhu19}
C.~Zhu, H.~Leung, S.~Hu, and Y.~Cai, ``{A Q-values Sharing Framework for
  Multiple Independent Q-learners},'' in \emph{AAMAS}, 2019, pp. 2324--2326.

\bibitem{Ilhan19}
E.~Ilhan, J.~Gow, and D.~Perez-Liebana, ``{Teaching on a Budget in Multi-Agent
  Deep Reinforcement Learning},'' in \emph{Proc. of IEEE Conference on Games},
  2019.

\bibitem{Fachantidis18}
A.~Fachantidis, M.~E. Taylor, and I.~Vlahavas, ``{Learning to Teach
  Reinforcement Learning Agents},'' \emph{ML and Knowledge Extraction}, vol.~1,
  pp. 21--42, 2018.

\bibitem{Gupta19}
V.~Gupta, D.~Anand, P.~Paruchuri, and B.~Ravindran, ``{Advice Replay Approach
  for Richer Knowledge Transfer in Teacher Student Framework},'' in \emph{Proc.
  of AAMAS}, 2019, pp. 1997--1999.

\bibitem{Vanhee19}
L.~Vanhee, L.~Jeanpierre, and A.-I. Mouaddib, ``{Augmenting Markong Decision
  Processes with Advising},'' in \emph{AAAI}, 2019, pp. 2531--2538.

\bibitem{Castro20}
P.~S. Castro, ``{Scalable Methods for Computing State Similarity in
  Deterministic Markov Decision Processes},'' in \emph{AAAI}, 2020, pp.
  10\,069--10\,076.

\bibitem{Watkins92}
C.~J. Watkins, ``{Q-Learning},'' \emph{Machine Learning}, 1992.

\bibitem{Zhu20}
T.~Zhu, D.~Ye, W.~Wang, W.~Zhou, and P.~S. Yu, ``{More Than Privacy: Applying
  Differential Privacy in Key Areas of Artificial Intelligence},'' \emph{IEEE
  Transactions on Knowledge and Data Engineering}, p. DOI:
  10.1109/TKDE.2020.3014246, 2020.

\bibitem{Ye20}
D.~Ye, T.~Zhu, S.~Shen, W.~Zhou, and P.~S. Yu, ``{Differentially Private
  Multi-Agent Planning for Logistic-like Problems},'' \emph{IEEE Transactions
  on Dependable and Secure Computing}, p. DOI: 10.1109/TDSC.2020.3017497, 2020.

\bibitem{Ye21}
D.~Ye, T.~Zhu, S.~Shen, and W.~Zhou, ``{A Differentially Private Game Theoretic
  Approach for Deceiving Cyber Adversaries},'' \emph{IEEE Transactions on
  Information Forensics and Security}, vol.~16, pp. 569--584, 2021.

\bibitem{Dwork14}
C.~Dwork and A.~Roth, ``{The Algorithmic Foundations of Differential
  Privacy},'' \emph{Foundations and Trends in Theoretical Computer Science},
  vol.~9, no. 3-4, pp. 211--407, 2014.

\bibitem{Kasiviswanathan2008531}
S.~Kasiviswanathan, H.~Lee, K.~Nissim, S.~Raskhodnikova, and A.~Smith, ``{What
  can we learn privately?}'' in \emph{Proc. of Annual IEEE Symposium on
  Foundations of Computer Science}, 2008, pp. 531--540.

\bibitem{Anton12}
H.~Anton, I.~Bivens, and S.~Davis, \emph{{Calculus}}.\hskip 1em plus 0.5em
  minus 0.4em\relax Wiley, 2012.

\bibitem{Wang08}
Y.~Wang and C.~W.~D. Silva, ``{A Machine-Learning Approach to Multi-Robot
  Coordination},'' \emph{Engineering Applications of Artificial Intelligence},
  vol.~21, no.~3, pp. 470--484, 2008.

\bibitem{Sakuma08}
J.~Sakuma, S.~Kobayashi, and R.~N. Wright, ``{Privacy-Preserving Reinforcement
  Learning},'' in \emph{Proc. of ICML}, 2008.

\end{thebibliography}
\end{document}